\documentclass[11pt]{article}
\usepackage{amsmath,amsthm,amssymb,cite,enumerate}
\usepackage[a4paper,left=0.5in,right=1in]{geometry}
\usepackage[usenames]{color}
\usepackage{graphicx}
\usepackage{epsfig}
\usepackage{dcolumn}
\usepackage{bm}
\usepackage{authblk} 

\usepackage{t1enc}	

\begin{document}

\def \d {{\rm d}}

\def \bm #1 {\mbox{\boldmath{$m_{(#1)}$}}}

\def \bF {\mbox{\boldmath{$F$}}}
\def \cF {\mbox{\boldmath{$\cal F$}}}
\def \bH {\mbox{\boldmath{$H$}}}
\def \bC {\mbox{\boldmath{$C$}}}
\def \bSS {\mbox{\boldmath{$S$}}}
\def \bS {\mbox{\boldmath{${\cal S}$}}}
\def \bV {\mbox{\boldmath{$V$}}}
\def \bff {\mbox{\boldmath{$f$}}}
\def \bT {\mbox{\boldmath{$T$}}}
\def \bk {\mbox{\boldmath{$k$}}}
\def \bl {\mbox{\boldmath{$\ell$}}}
\def \bn {\mbox{\boldmath{$n$}}}
\def \bbm {\mbox{\boldmath{$m$}}}
\def \tbbm {\mbox{\boldmath{$\bar m$}}}

\def \T {\bigtriangleup}
\newcommand{\msub}[2]{m^{(#1)}_{#2}}
\newcommand{\msup}[2]{m_{(#1)}^{#2}}

\newcommand{\be}{\begin{equation}}
\newcommand{\ee}{\end{equation}}

\newcommand{\beqn}{\begin{eqnarray}}
\newcommand{\eeqn}{\end{eqnarray}}
\newcommand{\AdS}{anti--de~Sitter }
\newcommand{\AAdS}{\mbox{(anti--)}de~Sitter }
\newcommand{\AAN}{\mbox{(anti--)}Nariai }
\newcommand{\AS}{Aichelburg-Sexl }
\newcommand{\pa}{\partial}
\newcommand{\pp}{{\it pp\,}-}
\newcommand{\ba}{\begin{array}}
\newcommand{\ea}{\end{array}}

\newcommand*\bR{\ensuremath{\boldsymbol{R}}}

\newcommand*\BF{\ensuremath{\boldsymbol{F}}}
\newcommand*\BR{\ensuremath{\boldsymbol{R}}}
\newcommand*\BS{\ensuremath{\boldsymbol{S}}}
\newcommand*\BC{\ensuremath{\boldsymbol{C}}}
\newcommand*\bg{\ensuremath{\boldsymbol{g}}}
\newcommand*\bE{\ensuremath{\boldsymbol{E}}}
\newcommand*\bA{\ensuremath{\boldsymbol{A}}}

\newcommand{\M}[3] {{\stackrel{#1}{M}}_{{#2}{#3}}}
\newcommand{\m}[3] {{\stackrel{\hspace{.3cm}#1}{m}}_{\!{#2}{#3}}\,}

\newcommand{\tr}{\textcolor{red}}
\newcommand{\tb}{\textcolor{blue}}
\newcommand{\tg}{\textcolor{green}}

\newcommand{\thorn}{\mathop{\hbox{\rm \th}}\nolimits}

\def\a{\alpha}
\def\g{\gamma}
\def\de{\delta}

\def\b{{\kappa_0}}

\def\E{{\cal E}}
\def\B{{\cal B}}
\def\R{{\cal R}}
\def\F{{\cal F}}
\def\L{{\cal L}}

\def\e{e}
\def\bb{b}

\newtheorem{theorem}{Theorem}[section] 
\newtheorem{cor}[theorem]{Corollary} 
\newtheorem{lemma}[theorem]{Lemma} 
\newtheorem{proposition}[theorem]{Proposition}
\newtheorem{definition}[theorem]{Definition}
\newtheorem{remark}[theorem]{Remark}

\title{Einstein-Maxwell fields as solutions of higher-order theories}

\author[1]{Marcello Ortaggio\thanks{ortaggio(at)math(dot)cas(dot)cz}}

\affil[1]{Institute of Mathematics of the Czech Academy of Sciences, \newline \v Zitn\' a 25, 115 67 Prague 1, Czech Republic}

\maketitle

\abstract{
We study four-dimensional Einstein-Maxwell fields for which any higher-order corrections to the field equations effectively reduces to just a rescaling of the gravitational and the cosmological constant. These configurations are thus simultaneous solutions of (virtually) any modified theory of gravity coupled (possibly non-minimally) to any electrodynamics.
In the case of non-null electromagnetic fields we provide a full characterization of such {\em universal solutions}, which correspond to a family of gravitational waves propagating in universes of the Levi-Civita--Bertotti--Robinson type. For null fields we first obtain a set of general necessary conditions, and then a full characterization for a special subfamily, which turns out to represent electromagnetic waves accompanied by pure radiation in the (anti-)Nariai background. The results are exemplified for the case of Born-Infeld, ModMax and Horndeski electrodynamics.
}

\vspace{.2cm}
\noindent

%

\tableofcontents

\section{Introduction}

\label{intro}

\subsection{Background and definition of universal and almost universal solutions}

\label{subsec_backgr}

Attempts to cure the divergent electron's self-energy motivated early proposals for modified classical theories of electrodynamics \cite{Mie12,Born33,BorInf34}. Nonlinear deviations from Maxwell's theory also arise in effective theories derived from quantum electrodynamics \cite{HeiEul36,Weisskopf36,Euler36,Schwinger51} or from string theory \cite{FraTse85}. While finding exact solutions of nonlinear electrodynamics (NLE) is in general much more difficult than in Maxwell's theory, Schr\"odinger observed \cite{Schroedinger35,Schroedinger43} that all null fields (defined by $F_{ab}F^{ab}=0={}^{*}\!F_{ab}F^{ab}$) which solve Maxwell's equations  automatically solve also a large set of NLE and are, in this sense, ``universal'' solutions. It was subsequently pointed out that plane waves (a special case of null fields) solve not only NLE but also higher-derivative theories \cite{Deser75}. Extensions beyond the case of plane waves and to allow for certain curved backgrounds, $p$-forms of higher ranks and higher dimensions have been obtained in recent years in the case of electromagnetic test fields \cite{OrtPra16,OrtPra18,HerOrtPra18}. However, it is clearly also desirable to understand how such universal electromagnetic fields backreact on the spacetime geometry, i.e., to study universal solutions of modified theories of gravity coupled to modified electrodynamics. In that direction, Kuchynka and the present author have obtained a characterization of Einstein-Maxwell solutions for which all higher-order correction vanish identically \cite{KucOrt19}. In this work we will go beyond the results of \cite{KucOrt19} by relaxing some of the assumptions made there, as we explain in the following.

In the absence of matter sources, an Einstein-Maxwell solution consists of a pair $(\bg,\bF)$ in which $\bg$ is a Lorentzian metric and $\bF=\d\bA$ a 2-form field, such that the Einstein-Maxwell field equations are satisfied, i.e.,
\beqn
 & & G_{ab}+\Lambda_0 g_{ab}=\kappa_0 T_{ab} , \label{Einst} \\ 
 & & \nabla_b{F}^{ab}=0 , \label{Maxw}
\eeqn
where $\Lambda_0$ and $\kappa_0$ are constants ($\kappa_0$ is dimensionless in geometrized units), $G_{ab}$ is the Einstein tensor and 
\be 
	T_{ab}=F_{ac}F_b^{\phantom{b}c}-\frac{1}{4}g_{ab}F_{cd}F^{cd} .
	\label{T}
\ee

We will be interested in modified theories of gravity coupled to an electromagnetic field, for which the field equations~\eqref{Einst}, \eqref{Maxw} are replaced by a more general system
\beqn
 & & G_{ab}+\Lambda g_{ab}=\kappa E_{ab} , \label{Einst_gen} \\ 
 & & \nabla_b{H}^{ab}=0 , \label{Maxw_gen}
\eeqn
where $\Lambda$ and $\kappa$ are again constants (possibly different from $\Lambda_0$ and $\kappa_0$), while $\bE$ is a symmetric, divergencefree 2-tensor and $\bH$ a 2-form, both constructed in terms of $\bF$, the Riemann tensor $\BR$ associated with $\bg$, and their covariant derivatives of arbitrary order (and contractions with $\bg$). In particular, if one considers theories defined by a Lagrangian density of the form ${\cal L}=\frac{1}{\kappa}(R-2\Lambda)+L(\bR,\nabla\bR,\ldots,\bF,\nabla\bF,\ldots)$, then $E_{ab}$ and $\nabla_b{H}^{ab}$ will be determined by the variations $\frac{1}{\sqrt{-g}}\frac{\delta S_L}{\delta g^{ab}}$ and $\frac{1}{\sqrt{-g}}\frac{\delta S_L}{\delta A_a}$, respectively, where $S_L=\int\d^4x\sqrt{-g}L$ (some explicit examples will be given in sections~\ref{sec_NLE} and \ref{sec_Horn}).

The purpose of the present paper is to identify a class of four-dimensional Einstein-Maxwell fields~$(\bg,\bF)$ solutions of~\eqref{Einst}, \eqref{Maxw} for which {\em all} possible tensors $\bE$ and $\bH$ that one can construct (as described above) are such that:
\begin{enumerate}
	\item\label{cond1} $\bE$ takes the form 
	\be
	 E_{ab}=b_1T_{ab}+b_2 g_{ab} ,
	 \label{univ_E}
	\ee
	where $b_1$ and $b_2$ are spacetime constants (but may depend on the particular $\bE$ being considered and on the specific solution $(\bg,\bF)$ chosen withing the given class) and $T_{ab}$ is as in~\eqref{T}. 
	\item\label{cond2} $\nabla_b{H}^{ab}=0$ identically.
\end{enumerate}
Pairs $(\bg,\bF)$ satisfying conditions~\ref{cond1} and \ref{cond2} will be referred to as {\em universal solutions}.

The reason for condition~\ref{cond2} is obvious -- the modified Maxwell equation~\eqref{Maxw_gen} will be automatically satisfied by the pair $(\bg,\bF)$. 
On the other hand, condition~\ref{cond1} ensures that the modified Einstein equation~\eqref{Einst_gen} is also satisfied for values of the coupling constants determined by
\be
 \Lambda-\Lambda_0=\kappa b_2 , \qquad \kappa_0=\kappa b_1 .
 \label{algebr_constr}
\ee

In other words, a universal Einstein-Maxwell solution $(\bg,\bF)$ does not only solve the Einstein-Maxwell theory~\eqref{Einst}, \eqref{Maxw} but also any modified theory admitting field equations of the form~\eqref{Einst_gen} and \eqref{Maxw_gen}, provided the algebraic conditions~\eqref{algebr_constr} are satisfied.\footnote{The spacetime constants $b_1$ and $b_2$ may be functions of $\Lambda_0$ and possible coupling constants of the modified theory, as well as of other parameters which characterize a specific pair $(\bg,\bF)$ -- see sections~\ref{sec_NLE} and \ref{sec_Horn} for some explicit examples (cf. also \cite{Coleyetal08,Gulluetal11,MalPra11prd,Kuchynkaetal19} in the vacuum case). It should be noted that in certain cases the algebraic equations~\eqref{algebr_constr} may not admit a real solution for $(\Lambda_0,\kappa_0)$ in terms of $(\Lambda,\kappa)$.
Moreover, there may be particular degenerate cases for which $b_1$ or $b_2$ become singular, or $b_1=0$, in which case the ``seed'' solution $(\bg,\bF)$ cannot be used to obtain a solution of the modified theory~\eqref{Einst_gen}, \eqref{Maxw_gen} with the method described above (cf. again the examples in sections~\ref{sec_NLE} and \ref{sec_Horn}).}  In the vacuum limit $T_{ab}=0$, condition~\eqref{univ_E} reduces to $E_{ab}=b_2 g_{ab}$, which defines universal spacetimes in the sens of \cite{Coleyetal08}, i.e., those which solve (virtually) any purely gravitational modified theory.

As we shall discuss in section~\ref{sec_4D_null}, in the case of null fields it will ultimately be necessary to relax condition~\ref{cond1} and allow $b_1$ in~\eqref{univ_E} to be a spacetime function. This effectively means that, from the viewpoint of the original Einstein-Maxwell theory, one needs to pick a pair $(\bg,\bF)$ that solves an Einstein equation also containing an additional term that can be interpreted as pure radiation (aligned with the null Maxwell field) -- i.e., one should add the quantitity $(\kappa b_1-\kappa_0)T_{ab}$ to the RHS of~\eqref{Einst} (ultimately affecting only one component of the Einstein equation since $\bF$ is null). Nevertheless, it should be emphasized no pure radiation will be present in the modified theory~\eqref{Einst_gen}, \eqref{Maxw_gen}, which thus remains electrovac. 
A similar approach in the case of modified gravities in vacuum was considered in \cite{Gulluetal11,MalPra11prd,Gursesetal13,Kuchynkaetal19} and for certain electrovac solutions in \cite{GurHeySen21,GurHeySen22}. Similarly as in \cite{Kuchynkaetal19}, solutions of this type can be referred to as {\em almost universal}. That is, an almost universal Einstein-Maxwell solution containing a null $\bF$ and aligned pure radiation solves any modified theory admitting field equations of the form~\eqref{Einst_gen} and \eqref{Maxw_gen}, provided the first of~\eqref{algebr_constr} is satisfied. Such solutions are thus only ``almost'' universal because $b_1$ may be a different function in different theories, which means that the pure radiation term on the Einstein-Maxwell side can be specified only once a particular modified electrovac theory has been chosen. This will be illustrated more explicitly in section~\ref{sec_4D_null} and, by an example, in section~\ref{subsec_Horn_null}

For Lagrangian theories that can be seen (in a sense specified in \cite{KucOrt19}) as higher-order modifications of the Einstein-Maxwell theory, the special configurations $(\bg,\bF)$ for which $b_1=1$ and $b_2=0$ have been fully characterized in \cite{KucOrt19} (in arbitrary dimension). These can be regarded as {\em strongly universal} solutions (using a terminology similar to \cite{Coleyetal08}) since $\Lambda=\Lambda_0$ and $\kappa=\kappa_0$, and corrections to the field equations vanish identically (rather than just being of the special form allowed by~\eqref{univ_E}). However, this leads to a rather restricted class of solutions, namely a specific $\Lambda_0=0$ subset of Kundt spacetimes of Petrov type III (or more special) coupled to a null $\bF$ and admitting a recurrent null vector field (see \cite{KucOrt19} for more details). The present paper will thus extend the analysis of \cite{KucOrt19} in various directions, in particular to include non-null electromagnetic fields, and also null fields in spacetimes of Petrov type~II and D, neither of which were covered by \cite{KucOrt19}.

\subsection{Preliminaries}

\label{subsec_prelim}

First, it is easy to see that~\eqref{univ_E} in condition~\ref{cond1} implies the two electromagnetic invariants 
\be
 I\equiv F_{ab}F^{ab} , \quad J\equiv {}^{*}\!F_{ab}F^{ab} ,
 \label{IJ}
\ee
are constant. This follows, for example, by considering the form of $E_{ab}$ which arises in NLE (given below in~\eqref{E_NLE}) for $L=I^q$ and $L=I+J^q$ ($q\neq0,1$).\footnote{A related result was obtained with a different method in theorem~3.2 of \cite{HerOrtPra18}.}

Next, since the energy-momentum tensor~\eqref{T} is traceless, it follows from the assumption~\eqref{univ_E} that all possible tensors $E_{ab}$ have constant trace (vanishing iff $b_2=0$). This applies, in particular, to tensors $E_{ab}$ constructed out of just the curvature tensor and its covariant derivatives. Thanks to theorem~3.2 of \cite{HerPraPra14}, one can thus conclude that metrics satisfying condition~\ref{cond1} must belong to the class of spacetimes for which all curvature invariants are constant (i.e., CSI spacetimes) \cite{ColHerPel06,ColHerPel09b}.

Furthermore, four-dimensional CSI spacetimes consist of a (proper) subset of {\em degenerate Kundt metrics}\footnote{Degenerate Kundt metrics \cite{ColHerPel09a,Coleyetal09} are defined as a subset of Kundt metrics for which the Riemann tensor and its covariant derivatives of arbitrary order are of algebraic type~II (or more special) aligned with the Kundt vector field $\bl$ (cf. also the review \cite{OrtPraPra13rev}). As it turns out, a Kundt spacetime is degenerate if it is of aligned Riemann type II and $\ell^a R_{,a}=0$ \cite{OrtPra16}. In the following we will also bear in mind the known fact that for Einstein-Maxwell fields in the Kundt class,  $\bl$ is automatically a principal null direction (PND) of $\bF$ and a double PND of the Weyl tensor (see, e.g., chapter~31 of \cite{Stephanibook}). This implies that in the Einstein-Maxwell case a Kundt space is necessarily degenerate and the Petrov type is II or more special.\label{footn_deg}} and of (locally) {\em homogeneous spacetimes} (theorem~2.1 of \cite{ColHerPel09b}). Taking advantage of this simplification, in the rest of the paper we can thus restrict our analysis to these two subclasses, without losing generality. In fact, for degenerate Kundt metrics it will be convenient to consider non-null and null electromagnetic fields separately. Furthermore, since a large amount of results for null fields in degenerate Kundt spacetimes of Petrov type~III, N and O has been already obtained \cite{KucOrt19} (see also \cite{GurHeySen22}), in the null case we shall focus on metrics of type~II (and D), leaving the complete analysis of the types III/N/O for future work. On the other hand, no restriction on the Petrov types will be a priori assumed neither in the non-null degenerate Kundt nor in the homogeneous case. The plan and main results of the paper are thus as follows:

\begin{enumerate}[(i)]

	\item section~\ref{sec_4D_degK_nonnull}: we show that the metric~\eqref{4D_nonnull_metric} with the 2-form~\eqref{4D_nonnull_F} represents the most general {\em universal} solution $(\bg,\bF)$ for which {\em $\bF$ is non-null} and {\em $\bg$ degenerate Kundt}

	\item section~\ref{sec_4D_null}: when {\em $\bF$ is null} and {\em $\bg$ degenerate Kundt}, we first obtain a set of general necessary conditions for $(\bg,\bF)$ to be (almost) universal (section~\ref{subsec_4D_null_gen}) and of Petrov type~II (D), and then a complete characterization of the special subfamily of {\em almost universal} solutions defined by the assumption $D\Psi_4=0$ (section~\ref{subsubsec_rec_spec}) -- this is represented by the metric~\eqref{rec_metric_spec} with the 2-form~\eqref{rec_Phi2}, and $H^{(0)}$ determined by~\eqref{Laplac_null} with \eqref{b1_null}

	\item	section~\ref{sec_homog}: both in the non-null and null cases, universal and almost universal solutions with a {\em homogeneous $\bg$} are shown to fall into the already investigated degenerate Kundt class, which also implies that spacetimes of Petrov type~I cannot occur (recall footnote~\ref{footn_deg}). 

\end{enumerate}

In sections~\ref{sec_NLE} and \ref{sec_Horn} we exemplify the previous results for the specific cases of NLE (in particular Born-Infeld and ModMax theories) and Horndeski's theory, respectively. Finally, in the appendices we prove some technicalities needed in the main body of the paper. To that end, we employ the formalism of Geroch-Held-Penrose \cite{GHP}, which is briefly summarized in appendix~\ref{subsec_GHP}. Then we analyze separately the case of non-null and null fields in appendices~\ref{app_nonnull} and \ref{app_null}, respectively. Some of the results obtained there go beyond the purposes of the present paper and can be useful also for future applications.

\paragraph{Notation}

$\BR$, $\BC$, $\BS$ denote the Riemann and Weyl tensors and the tracefree part of the Ricci tensor, respectively. A real 2-form is denoted by $\BF$ or $\bH$, while a generic spinor (possibly complex) by $\bS$. We will use the abbreviations ``b.w.'' and ``s.w.'' for boost and spin weight, respectively -- these are the standard notions of the Newman-Penrose (NP) \cite{NP} and Geroch-Held-Penrose (GHP) \cite{GHP} formalisms (cf. also \cite{Stephanibook,penrosebook1} and appendix~\ref{subsec_GHP}). An asterisk will denote Hodge duality. A tensor is called a VSI (\underline{v}anishing \underline{s}calar \underline{i}nvariants) tensor if all its scalar polynomial invariants (obtained by contracting polynomials in the metric, the tensor itself and its covariant derivatives of arbitrary order) vanish. A similar terminology is employed for tensors possessing only \underline{c}onstant \underline{s}calar \underline{i}nvariants, referred to as CSI tensors. When a metric or a spacetime is called VSI [CSI], it is actually meant that its Riemann tensor is VSI [CSI].

To express the Maxwell field, it will be often convenient to use the self-dual 2-form \cite{Stephanibook,penrosebook1}
\be
 \mathcal{F}_{ab}=F_{ab}+i{}^{*}\!F_{ab} . \label{Fself}
\ee
Its algebraic complex invariant will be hereafter assumed (by the comments in the paragraph following~\eqref{IJ}) to be constant, i.e.,
\be
 \mathcal{F}_{ab}\mathcal{F}^{ab}=2(I+iJ)=\mbox{const} . 
 \label{complex_inv}
\ee
In terms of $\cF$, the energy-momentum tensor~\eqref{T} takes the compact form
\be
 T_{ab}=\frac{1}{2}\mathcal{F}_{ac}\bar{\mathcal{F}}_b^{\phantom{b}c} . 
 \label{T_2}
\ee

We will be using the complex NP formalism \cite{NP}, with the conventions of \cite{Stephanibook}. In a complex frame $(\bl,\bn,\mbox{\boldmath{$m$}},\mbox{\boldmath{$\bar m$}})$, the metric reads 
\be
	g_{ab}=2m_{(a}\bar m_{b)}-2\ell_{(a}n_{b)} ,
	\label{g_tetrad}
\ee
while the form of $\cF$ and of the Maxwell equation~\eqref{Maxw} will be given in sections~\ref{sec_4D_degK_nonnull} and \ref{sec_4D_null} in the non-null and null case, respectively. In the above frame, the tracefree part of the Einstein equation~\eqref{Einst} reads $\Phi_{ij}=\kappa_0\Phi_{i}\bar\Phi_{j}$ ($i,j=0,1,2$), while its trace gives $R=4\Lambda_0$.

\section{Degenerate Kundt: non-null fields}

\label{sec_4D_degK_nonnull}

\subsection{Necessary conditions}

In a frame adapted to the two null PNDs of $\F$, one has \cite{Stephanibook}
\be
 \mathcal{F}_{ab}= 4\Phi_1 (m_{[a}\bar m_{b]}-\ell_{[a}n_{b]}) ,
 \label{F1_0}
\ee 
so that $\bl$ and $\bn$ are repeated PNDs of the energy-momentum tensor
\be
 T_{ab}= 4\Phi_1\bar\Phi_1(m_{(a}\bar m_{b)}+\ell_{(a}n_{b)}) ,
 \label{T_nonnull}
\ee 
and the only non-zero component of the traceless Ricci tensor is $\Phi_{11}=\kappa_0\Phi_1\bar\Phi_1$. Since $\mathcal{F}_{ab}\mathcal{F}^{ab}=-16\Phi_1^2$, by~\eqref{complex_inv} one gets
\be
 \Phi_1=\mbox{const}\neq0 .
 \label{Phi1_const}
\ee

At least one of the PNDs of $\cF$ is Kundt (footnote~\ref{footn_deg}) -- for definiteness, let's say it is $\bl$. We thus have
\be
  \Psi_0=\Psi_1=0 , \qquad \kappa=\rho=\sigma=\epsilon=0 , 
	\label{degK}
\ee
where the last equality can be ensured by exploiting a freedom of boosts and spins, without loss of generality.

With the above conditions, Maxwell's equation reduce to (cf., e.g., \cite{Stephanibook}) 
\be
 \pi=\tau=\mu=0 .
\label{ptm}
\ee
This implies that $\bl$ is recurrent and that the frame in use is parallelly transported along it. A further rescaling enables one to set also (without affecting the previous conditions)
\be
  \alpha+\bar\beta=0 , 
	\label{ab}
\ee
such that $\bl$ becomes a gradient.

By condition~\eqref{univ_E} with~\eqref{T_nonnull}, the following symmetric 2-tensor\footnote{A tensor of this form appears, for example, in the Einstein equation of Horndeski's electrodynamics, cf.~\eqref{E_Horn} below, as can be seen using the identity (recall~\eqref{Fself}) $\nabla_d\,{}^{*}\!F_{ac}\nabla^c\,{}^{*}\!F_{b}^{\ d}=\frac{1}{4}(\nabla_d\mathcal{F}_{ac}\nabla^c\mathcal{\bar F}_{b}^{\ d}-\nabla_d\mathcal{F}_{ac}\nabla^c\mathcal{F}_{b}^{\ d}+\mbox{c.c.})$. We note that here $\nabla_d\mathcal{F}_{ac}\nabla^c\mathcal{F}_{b}^{\ d}=0$ (thanks to~\eqref{F1_0}, \eqref{degK}, \eqref{ptm}), and the additional term $F^{ce}F^{d}_{\phantom{d}e}\,{}^{*}\!R^*_{acbd}$ contained in~\eqref{E_Horn} does not affect the present discussion since it does not possess a component proportional to $\ell_a \ell_b$, as can be seen using~\eqref{F1_0}, the fact that the traceless Ricci tensor must obey $\Phi_{11}=\kappa_0\Phi_1\bar\Phi_1$ (thus possessing only components of b.w.~0) and ${}^{*}\!R^*_{acbd}=-C_{acbd}+\ldots$, where the dots denote terms constructed from the Ricci tensor and the metric.\label{footn_Hornd}}
\be
	\nabla_d\mathcal{F}_{(a|c}\nabla^c\mathcal{\bar F}_{|b)}^{\ d}=16|\Phi_1|^2 |\lambda|^2\ell_a \ell_b ,
	\label{dFdbF}
\ee	
must vanish in order for $\F$ to be universal, so that
\be
 \lambda=0 .
\label{l}
\ee
We further note that the Ricci and Bianchi identities give \cite{Stephanibook}
\be
 \Psi_2=-\frac{\Lambda_0}{3} , \qquad \Psi_3=0 , \qquad D\Psi_4 =0 , \qquad D\nu=0 . 
 \label{nonnull_NP}
\ee

We have now enough information to introduce adapted coordinates and arrive at the explicit form of the line element. We already observed that $\bl$ is a gradient, while eqs.~\eqref{degK}, \eqref{ptm}, \eqref{l} and \eqref{ab} further ensure $\bbm\wedge\d\bbm=0$, $\bl\wedge\d\bn=0$, $\bl\wedge\d\bbm=-2\beta\bl\wedge\bbm\wedge\tbbm$. We can thus define coordinates $(u,r,\zeta,\bar\zeta)$ such that 
\beqn
  & & \bl=-\d u , \qquad \bbm=P^{-1}\d \bar\zeta , \qquad \bn=-(\d r+W\d\zeta+\bar W\d \bar\zeta+H\d u) , \\
	& & P_{,r}=0 , \qquad W_{,r}=0 , \qquad W_{,\bar\zeta}-\bar W_{,\zeta}=0 . \label{PW}
\eeqn
An $r$-independent spin can be used to make $P$ real (without affecting the previously obtained conditions), while imposing $\mu=0$ (cf.~\eqref{ptm}) additionally gives
\be
 P_{,u}=0 .
\ee 
The latter condition also ensures that Maxwell's equation is satisfied. Thanks to the last of \eqref{PW}, a coordinate transformation $r\mapsto r+g(u,\zeta,\bar\zeta)$ can be used to set 
\be
	W=0 .
\ee

We have thus arrived at a special subcase of the canonical Kundt line-element (cf. \cite{Stephanibook} for more details). The component $(\zeta\bar\zeta)$ of Einstein's equations gives
\be
	H=-\frac{1}{2}k_1r^2+rH^{(1)}(u,\zeta,\bar\zeta)+H^{(0)}(u,\zeta,\bar\zeta) , \qquad k_1=\Lambda_0-2\kappa_0\Phi_1\bar\Phi_1 .
	\label{H_nonnull}
\ee
The component $(ur)$ reads $\T\ln P=\Lambda_0+2\kappa_0\Phi_1\bar\Phi_1$ (where $\T=2P^2\pa_{\bar\zeta}\pa_\zeta$ is the Laplace operator in the transverse 2-space spanned by $(\zeta,\bar\zeta)$), which enables one to redefine the coordinates $(\zeta\bar\zeta)$ such that
\be
 P=1+\frac{k_2}{2}\zeta\bar\zeta , \qquad k_2=\Lambda_0+2\kappa_0\Phi_1\bar\Phi_1 .
 \label{P_nonnull}
\ee
The equation $(u\zeta)$ gives $H^{(1)}=H^{(1)}(u)$, which guarantees that we can redefine $u\mapsto U(u)$, $r\mapsto r/\dot U$ to achieve
\be
 H^{(1)}=0 .
\ee
Lastly, the equation $(uu)$ requires $\T H^{(0)}=0$ and therefore
\be
 H^{(0)}=h(u,\zeta)+\bar h(u,\bar \zeta) .
\ee

The line-element thus finally reads
\be
 \d s^2=-2\d u\d r-2\left[-\frac{1}{2}k_1r^2+h(u,\zeta)+\bar h(u,\bar \zeta)\right]\d u^2+\frac{2\d \zeta\d \bar\zeta}{\left(1+\frac{k_2}{2}\zeta\bar\zeta\right)^2} ,
 \label{4D_nonnull_metric}
\ee
and the Maxwell field~\eqref{F1_0} becomes
\be
 \cF= 2\Phi_1 \left[\frac{\d \bar\zeta\wedge\d \zeta}{\left(1+\frac{k_2}{2}\zeta\bar\zeta\right)^2}-\d u\wedge\d r\right] ,
 \label{4D_nonnull_F}
\ee 
together with~\eqref{Phi1_const} and (from \eqref{H_nonnull}, \eqref{P_nonnull}), $k_1=\Lambda_0-2\kappa_0\Phi_1\bar\Phi_1$, $k_2=\Lambda_0+2\kappa_0\Phi_1\bar\Phi_1$. The relative signs of $k_1$, $k_2$ and $\Lambda_0$ must be such that $\Phi_1\bar\Phi_1>0$ \cite{Bertotti59,OrtPod02}.

\subsection{Sufficiency of the conditions}

\label{subsec_nonnull_suff}

In the above section we have obtained a set of necessary conditions for an Einstein-Maxwell solution to be universal, which led to the solution~\eqref{4D_nonnull_metric}, \eqref{4D_nonnull_F}. Let us now show that those conditions are also sufficient, i.e., that the solution~\eqref{4D_nonnull_metric}, \eqref{4D_nonnull_F} is indeed {\em universal}.

First, we observe that the curvature tensor contains only components of b.w.~0 (the Weyl $\Psi_2$ and the Ricci $R=4\Lambda_0$ and $\Phi_{11}=\kappa_0\Phi_1\bar\Phi_1$) and $-2$ (i.e., $\Psi_4$), while $\cF$ only components of b.w.~0 (i.e., $\Phi_1$) -- in both cases, those of b.w.~0 are constant. Furthermore, the covariant derivatives of $\cF$ and of the energy-momentum tensor (which is proportional to the tracefree part of the Ricci tensor) are of the form 
\be
 \nabla_c\mathcal{F}_{ab}=8\nu\Phi_1\ell_c\ell_{[a}m_{b]} , \qquad \nabla_c T_{ab}=-8\Phi_1\bar\Phi_1\ell_c\ell_{(a}(\nu m_{b)}+\bar\nu\bar m_{b)}) . 
 \label{der_F_T}
\ee
Together with~\eqref{Phi1_const} and \eqref{nonnull_NP}, this means that both tensors \eqref{der_F_T} are {\em 1-balanced} tensors (as defined in \cite{HerPraPra14}, cf. also appendix~\ref{app_nonnull}). One can similarly show that also the covariant derivative of the Weyl tensor is 1-balanced (a proof of this statement can be found in section~4.1 of \cite{HerPraPra17}), therefore the covariant derivative of the full Riemann tensor is 1-balanced as well. This implies that covariant derivatives of arbitrary order of both $\bR$ and $\cF$ possess only components of b.w.~$\le-2$ (see lemma~A.7 of \cite{HerOrtPra18} and the proof of proposition~2.9 in the same reference, and \cite{Pravdaetal02,Coleyetal04vsi,HerPraPra14,Herviketal15,OrtPra16,HerPraPra17} for several related earlier results). 
Therefore, any possible $E_{ab}$ can only contain terms of b.w.~0 and $-2$, while possible $H_{ab}$ can only be of b.w.~0 (since a 2-form cannot have components of b.w.~$\le-2$).

All the components of b.w.~0 are the same as those of the ``background'' direct-product spacetime defined by setting $h=0$ in \eqref{4D_nonnull_metric} (cf. \cite{ColHerPel10,HerOrtPra18}) and therefore are invariant under the symmetries of the latter (cf. \cite{HerOrtPra18} for related comments).
This means that components of b.w.~0 of any possible tensor that can be constructed out of $\bR$ and $\cF$ (their covariant derivatives cannot contribute since have b.w. $\le-2$) will still admit the same symmetries -- in particular, both their boost and spin weights will be zero. Therefore, the b.w.~0 part of any possible symmetric 2-tensor $E_{ab}$ will be given by a linear combination, with constant coefficients, of $g_{ab}$ and $T_{ab}$ (in agreement with~\eqref{univ_E}), while any possible 2-form (necessarily of b.w.~0, as observed above) will consist of a linear combination of the 2-volume elements of the two factor spaces. It is easy to check that any such 2-form is necessarily closed and co-closed \cite{HerOrtPra18}, therefore the generalized Maxwell equation~\eqref{Maxw_gen} is also automatically satisfied.

We can thus hereafter focus only on the possible b.w.~$-2$ components of $E_{ab}$. These are proportional to $\ell_a\ell_b$ and thus have s.w.~0. Let us first note that the b.w.~$-2$ part of $\bR$ cannot contribute to those, since it is traceless and has s.w.~$\mp2$ (thus any symmetric 2-tensor obtained by contractions of the b.w.~$-2$ part of $\bR$ with the b.w.~0 parts of $\bR$ and $\cF$ is necessarily zero -- cf. \cite{Coleyetal08} for related comments). Thanks to this and to the previous observations, b.w.~$-2$ components of $E_{ab}$ must thus contain terms linear in the covariant derivatives of either $\bR$ or $\cF$, suitably contracted with certain tensor components of b.w.~0 and s.w.~0. Therefore, only terms of the covariant derivatives of $\bR$ and $\cF$ possessing b.w.~$-2$ {\em and} s.w.~0 can give rise to b.w.~$-2$ components of $E_{ab}$. However, one can show iteratively that such terms vanish for covariant derivatives of arbitrary order, and thus the b.w.~$-2$ components of $E_{ab}$ are identically zero. Such a proof can be found in appendix~\ref{app_appl_nonnull}.

To summarize, we have shown that the family of Einstein-Maxwell solutions $(\bg,\bF)$ given by~\eqref{4D_nonnull_metric} and \eqref{4D_nonnull_F} are universal in the sense defined in section~\ref{subsec_backgr}. It also follows from the above derivation that they are the unique universal solutions when $\cF$ is non-null.

Let us recall that in the Einstein-Maxwell theory solutions~\eqref{4D_nonnull_metric}, \eqref{4D_nonnull_F} represent non-expanding gravitational waves propagating in the Levi-Civita--Bertotti--Robinson, charged (anti-)Nariai and Pleba\'nski-Hacyan direct product universes \cite{LeviCivita17BR,Nariai51,Bertotti59,Robinson59,CahDef68,PlebHac79}, to which they reduce for $h=0$. They were first found in \cite{GarAlvar84,Khlebnikov86} and further studied in \cite{OrtPod02,PodOrt03,OrtAst18} (see also \cite{Lewand92,Ortaggio02} in the vacuum limit $\Phi_1=0$). They are of Petrov type II(D) iff $\Lambda_0\neq0$ (i.e., iff $k_1\neq-k_2$) and of type N(O) otherwise. Since for these solutions the covariant derivatives of $\cF$ and $\bR$ are both 1-balanced (as noticed above), they cannot be used to construct any invariants. Furthermore, the frame components of b.w.~0 of both $\bR$ and $\cF$ are constant ($R=4\Lambda_0$ and \eqref{Phi1_const} with the first of \eqref{nonnull_NP}), which enables one to conclude that no non-constant invariants can be constructed, neither from $\bR$ nor from $\cF$ (not even mixed ones) -- in particular demonstrating that the metric is CSI (cf. section~\ref{subsec_prelim}).

\section{Degenerate Kundt (Petrov type II): null fields}

\label{sec_4D_null}

\subsection{On the general class}

\label{subsec_4D_null_gen}

In this case in an adapted frame one has \cite{Stephanibook}
\be
 {\mathcal F}_{ab}=4\Phi_2 \ell_{[a}m_{b]} ,
\label{F_null}
\ee
giving
\be
 T_{ab}= 2\Phi_2\bar\Phi_2\ell_a\ell_b ,
 \label{T_null_gen}
\ee 
where the unique PND $\bl$ of $\cF$ is necessarily Kundt and a multiple PND of the Weyl tensor (cf. footnote~\ref{footn_deg}), i.e., 
\be
  \Psi_0=0=\Psi_1 .
	\label{null_bw21}
\ee
In a parallelly transported frame adapted to $\bl$, one has
\be
 \kappa=\rho=\sigma=\pi=\epsilon=0 .
 \label{degK_null}
\ee

Maxwell's equation thus takes the form \cite{Stephanibook} 
\be
 D\Phi_2=0 , \qquad \delta\Phi_2=(\tau-2\beta)\Phi_2 .
\ee

It follows that $\cF$ is a {\em balanced} tensor (as defined in \cite{Pravdaetal02,Coleyetal04vsi}, cf. also appendix~\ref{app_null}) together with its covariant derivatives of arbitrary order, and therefore it is VSI \cite{OrtPra16}. 

Moreover, by the assumptions (cf. section~\ref{subsec_prelim}), the metric is VSI iff $\Lambda_0=0=\Psi_2$ \cite{Pravdaetal02}, CSI otherwise -- in both cases, $\Psi_2=$const. From now on we will restrict ourselves to spacetimes of Petrov type~II and D, i.e., $\Psi_2\neq0$ will be assumed hereafter. Partial results for the types III, N and O can be found in \cite{KucOrt19,GurHeySen22}.

In a Kundt CSI spacetime of Weyl type II and traceless-Ricci type N, a subset of the Ricci and Bianchi identities gives \cite{Stephanibook} (using \eqref{null_bw21}, \eqref{degK_null})
\beqn
 & & D\gamma=\tau\alpha+\bar\tau\beta+\Psi_2-R/24 , \qquad D\nu=\bar\tau\mu+\tau\lambda+\Psi_3 , \qquad D\mu=\Psi_2+R/12 , \label{Ricci4D_2} \\
 & & D\alpha=0=D\beta , \qquad D\lambda=0 , \qquad \bar\delta\tau=-(\bar\beta-\alpha-\bar\tau)\tau+\Psi_2+R/12 \label{Ricci4D} \\
 & & \tau\Psi_2=0  , \qquad \bar\delta\Psi_3-D\Psi_4=-2\alpha\Psi_3+3\lambda\Psi_2 , \qquad D\Psi_3=0  , \qquad D\Phi_{22}=0 . \label{Bianchi4D}
\eeqn

The latter equation suffices to show that $\BS$ and its covariant derivatives of arbitrary order are 1-balanced  (cf. also lemma~B.4 of \cite{KucOrt19} for a more general result). Using also the commutator $[\delta,D]=(\bar\alpha+\beta)D$ \cite{Stephanibook}, from the second of \eqref{Bianchi4D} one finds
\be
 D^2\Psi_4=0 .
 \label{D2Psi4}
\ee

Since for type II one has $\Psi_2\neq0$, the first of \eqref{Bianchi4D} implies
\be
 \tau=0 ,
\label{tau=0}
\ee 
which means that $\bl$ is recurrent.

Then \eqref{Ricci4D_2}--\eqref{Bianchi4D} further give
\be
 \Psi_2+R/12=0 , \qquad D^2\gamma=0 , \qquad D^2\nu=0 , \qquad D\mu=0 ,
 \label{null_typeII}
\ee
the first of which also reads alternatively $\Psi_2=-\Lambda_0/3\neq0$.

Following the same steps as in \cite{HerPraPra17} one can therefore argue that $\nabla^k\BC$ is balanced for any $k\ge1$ (while $\BC$ is not, since of type II). Furthermore, we note that for any $k\ge1$ $\nabla^k\BC$ is 1-balanced (and thus $\nabla^k\bR$ is 1-balanced since $\BS$ also is, see above) iff $D\Psi_4=0$, which thus characterizes a geometrically privileged subcase (to be analyzed in section~\ref{subsubsec_rec_spec}).

Recalling that the spacetime in question must be recurrent Kundt and solve the Einstein-Maxwell equations sourced by~\eqref{F_null} (and thus necessarily be aligned with $\bl$ \cite{Stephanibook}), the metric can be written as \cite{Walker50,LerMcL73,Stephanibook}
\be
 \d s^2=-2\d u\left(\d r+W\d\zeta+\bar W\d\bar\zeta+H\d u\right)+2P^{-2}\d\zeta\d\bar\zeta ,
\label{rec_metric}
\ee
with \cite{LerMcL73}\footnote{A different gauge choice is possible such that $H^{(1)}=0$ (at the price of changing the form of $W$), cf. \cite{Lewand92,Kadlecovaetal09}.}
\beqn
  & & W=P^{-2}\bar g(u,\bar\zeta) , \qquad  H=-r^2\frac{\Lambda_0}{2}+rH^{(1)}+H^{(0)} , \qquad P=1+\frac{\Lambda_0}{2}\zeta\bar\zeta \label{W_rec} \\
	& & 2H^{(1)}=-\Lambda_0P^{-1}\left(\zeta\bar g+\bar\zeta g\right)+g_{,\zeta}+\bar g_{,\bar \zeta} , \label{H1_rec}
\eeqn
and the Maxwell field~\eqref{F_null} is given by
\be
 \cF=-2\bar f(u,\bar\zeta)\d u\wedge\d\bar\zeta , 
 \label{rec_Phi2}
\ee
where the complex functions $g$ and $f$ are arbitrary. In a null tetrad \cite{Stephanibook} 
\be
 \bl=-\d u , \qquad \bbm=P^{-1}\d \bar\zeta-PW\d u , \qquad \bn=-\d r-(H+P^2W\bar W)\d u , 
\ee
eq.~\eqref{rec_Phi2} corresponds to~\eqref{F_null} with 
\be
 \Phi_2=P\bar f(u,\bar\zeta) .
\ee

We also note that the form~\eqref{W_rec}  of $W$ implies 
\be
 \lambda=0 .
 \label{lambda=0}
\ee

Eqs.~\eqref{rec_metric}--\eqref{rec_Phi2} ensure that all components of the Einstein equation~\eqref{Einst} are satisfied, except for the one along $\ell_a\ell_b$, i.e., the component $\Phi_{22}=\kappa_0\Phi_2\bar\Phi_2$ in NP notation -- hence at this stage also pure radiation aligned with $\cF$ is present. If one imposes also $\Phi_{22}=\kappa_0\Phi_2\bar\Phi_2$, then the real function $H^{(0)}=H^{(0)}(u,\zeta,\bar\zeta)$ must obey a second order linear PDE that in GHP notation can be written compactly as (using~\eqref{degK_null}, \eqref{tau=0}, \eqref{lambda=0})\footnote{Hereafter the symbols $\thorn$, $\thorn'$, $\eth$ and $\eth'$ denote the standard derivative operators defined in the compact GHP formalism \cite{GHP}, which is reviewed in appendix~\ref{subsec_GHP}. The explicit form of~\eqref{Einst-2_null} in terms of ordinary partial derivatives, which will not be needed in the rest of the paper, can be found in \cite{LerMcL73,Stephanibook}. For the sake of definiteness, let us only mention here that, in the spacetime~\eqref{rec_metric}--\eqref{H1_rec}, the function $H^{(0)}$ appears in \eqref{Einst-2_null} only through the term $\eth\nu=\frac{1}{2}\T H^{(0)}+\ldots$, where where $\T=2P^2\pa_{\bar\zeta}\pa_\zeta$ is the Laplace operator in the transverse 2-space spanned by $(\zeta,\bar\zeta)$, and the ellipsis denotes quantities which vanish when $g=0$ (as will be relevant in section~\ref{subsubsec_rec_spec}).\label{foot_GHP}} 
\be
 \eth\nu-\thorn'\mu-\mu^2=\kappa_0\Phi_2\bar\Phi_2 .
 \label{Einst-2_null}
\ee
However, in the following we will {\em not} impose~\eqref{Einst-2_null} and thus $H^{(0)}$ will remain, for now, unconstrained. The reason for this is related to the comments about almost universal solution given in section~\ref{subsec_backgr}, as will be made more explicit in section~\ref{subsubsec_rec_spec}.

By setting $g=f=H^{(2)}=0$ in \eqref{rec_metric}--\eqref{rec_Phi2} one obtains the (anti)-Nariai vacuum ``background'' \cite{Nariai51}, which is a direct product of two 2-dimensional spaces with identical Gaussian curvature and therefore \cite{Herviketal15} a universal spacetime. The b.w.~0 part of the curvature tensor (with components $\Psi_2$ and $R$) of the full metric \eqref{rec_metric} is the same as the one of the (anti)-Nariai background (since it does not contain the functions $W$, $H^{(1)}$ and $H^{(2)}$, as noted in \cite{ColHerPel10}). By the result of \cite{Herviketal15} just mentioned, this means that components of $E_{ab}$ of b.w.~0 reduce to terms proportional to the metric and are thus harmless. It also follows that no components of $H_{ab}$
of b.w.~0 (which could only come from the curvature tensor, since $\cF$ is balanced) are possible -- otherwise the ``square'' of such terms would produce a symmetric 2-tensor not proportional to the metric, contradicting \cite{Herviketal15}. One can thus focus on components of $E_{ab}$ and $H_{ab}$ of negative b.w.. For simplicity, in the rest of this section we will restrict ourselves to the special subcase (identified above) with $D\Psi_4=0$. We recall that even in the vacuum case a conclusive answer for the generic case $D\Psi_4\neq0$ has not been obtained yet \cite{HerPraPra17} -- this will thus deserve a separate investigation.

\subsection{Special subclass $D\Psi_4=0$}

\label{subsubsec_rec_spec}

\subsubsection{Necessary conditions}

\label{subsubsec_necessary_null}

From now on we thus assume the additional condition $D\Psi_4=0$. It is not difficult to see \cite{LerMcL73} that for the metric \eqref{rec_metric}--\eqref{H1_rec} one has $D\Psi_4=0\Leftrightarrow g_{,\zeta\zeta\zeta}=0$, so that here we can take
\be
 g=a_0(u)+a_1(u)\zeta+a_2(u)\zeta^2 . 
\ee

With \eqref{W_rec} this gives
\be 
	\left(1+\frac{\Lambda_0}{2}\zeta\bar\zeta\right)^3(W_{,\bar\zeta}-\bar W_{,\zeta})=\left(1-\frac{\Lambda_0}{2}\zeta\bar\zeta\right)(\bar a_1-a_1)+\bar\zeta(2\bar a_2+\Lambda_0 a_0)-\zeta(2a_2+\Lambda_0\bar a_0) .
\ee	

Next, a transformation of the form\footnote{This is discussed thoroughly in \cite{LerMcL73} in the context of type~D spacetimes. Note, however, that this part of the analysis of \cite{LerMcL73} applies also to spacetimes of type~II provided they satisfy $D\Psi_4=0$, therefore there is no need to repeat here all the details.} 
\be
 \zeta'=\frac{c(u)\zeta+d(u)}{\bar c(u)-\frac{\Lambda_0}{2}\bar d(u)\zeta}  \qquad \left(c\bar c+\textstyle{\frac{\Lambda_0}{2}}d\bar d\neq0\right) ,
\label{Moebius}
\ee
under which $\left(1+\frac{\Lambda_0}{2}\zeta\bar\zeta\right)^{-2}\d\zeta\d\bar\zeta=\left(1+\frac{\Lambda_0}{2}\zeta'\bar\zeta'\right)^{-2}\d\zeta'\d\bar\zeta'$, can be used to set $\bar a_1'-a_1'=0$ and $2\bar a_2'+\Lambda_0 a_0'=0$. In the new coordinates one thus has (after dropping the primes) $W=\varphi_{,\zeta}$, where $\varphi=\left(1+\frac{\Lambda_0}{2}\zeta\bar\zeta\right)^{-1}(\bar a_0\zeta+a_0\bar\zeta+a_1\zeta\bar\zeta)$ is a real function. This enables one to make a further transformation $r'=r+\varphi$ to set $W'=0$, i.e., $g'=0$ \cite{LerMcL73,Stephanibook}. This also gives $H'^{(1)}=0$ (recall \eqref{H1_rec}) and the final form of the metric is thus (dropping again the primes)
\be
 \d s^2=-2\d u\d r+\left(r^2\Lambda_0-2H^{(0)}\right)\d u^2+\frac{2\d\zeta\d\bar\zeta}{\left(1+\frac{\Lambda_0}{2}\zeta\bar\zeta\right)^2} ,
\label{rec_metric_spec}
\ee
with \eqref{rec_Phi2} unchanged (up to suitably redefining $f$ after \eqref{Moebius}).

Let us also note that one now has  	
\be 
	\mu=0 , \qquad \Psi_3=0 ,
\ee	
which will be useful in the following (these conditions can be obtained from the general expressions given in \cite{Stephanibook}).

As mentioned in the paragraph below~\eqref{lambda=0}, an Einstein-Maxwell solution without pure radiation would correspond to determining $H^{(0)}$ from \eqref{Einst-2_null}. For $W=0=H^{(1)}$ the latter simply becomes $\T H^{(0)}=2\kappa_0\Phi_2\bar\Phi_2$, whose solution reads \cite{PodOrt03}
\be
  H^{(0)}=\kappa_0\int f\d\zeta\int\bar f\d\bar\zeta+ h(u,\zeta)+\bar h(u,\bar \zeta) ,
 \label{H0_null}
\ee
with $h(u,\zeta)$ arbitrary. We will, however, allow for arbitrary aligned pure radiation in the Einstein-Maxwell theory and thus will {\em not} assume~\eqref{H0_null} but rather keep $H^{(0)}$ unspecified at this stage (it will be fixed in a theory-dependent way by~\eqref{Laplac_null}, see the comments following it). In passing, let us note that for the metric~\eqref{rec_metric_spec}, apart from $\Psi_2=-\Lambda_0/3$, the only other Weyl component is given by $\Psi_4=P(PH^{(0)})_{,\bar\zeta\bar\zeta}$.

\subsubsection{Sufficiency of the conditions}

\label{subsubsec_null_suff}

We have obtained above a set of necessary conditions for an Einstein-Maxwell solution with a null field and $D\Psi_4=0$ to be universal or almost universal, which led to the solution~\eqref{rec_metric_spec}, \eqref{rec_Phi2}. To be precise, the latter is an Einstein-Maxwell solution only once $H^{(0)}$ is taken as in \eqref{H0_null} -- for a generic $H^{(0)}$, an additional term representing aligned pure radiation is also present (as already mentioned in section~\ref{subsec_backgr}). We shall now argue that the necessary conditions obtained in section~\ref{subsubsec_necessary_null} are also sufficient for \eqref{rec_metric_spec}, \eqref{rec_Phi2} to describe an {\em almost universal} solution. We emphasize again that a difference with respect to the universal solutions of section~\ref{sec_4D_degK_nonnull} is that here the function $H^{(0)}$ cannot be be fixed a priori in the Einstein-Maxwell theory (i.e., as in \eqref{H0_null}) but can be specified only once a particular theory has been chosen (whereby the solution \eqref{rec_metric_spec}, \eqref{rec_Phi2} is only ``almost'' universal -- cf. \cite{Gulluetal11,MalPra11prd,Kuchynkaetal19,GurHeySen21,GurHeySen22} for a similar approach in other contexts).

As noticed at the end of section~\ref{subsec_4D_null_gen}, we only need study components of negative b.w.. Since here $\nabla^k\bR$ is 1-balanced (for $k\ge1$) and $\cF$ is balanced (cf. section~\ref{subsec_4D_null_gen}), components of $\bH$ can only be constructed linearly in terms of $\cF$ and its covariant derivatives (and the complex conjugates of those). As it turns out, however, any such term necessarily has $s\neq0$ if $b=-1$ (this is proven in appendix~\ref{app_appl_null}), which means that the divergence of such terms in necessarily zero (since a vector with $b=-1$ can only have $s=0$) and therefore $\nabla_b{H}^{ab}=0$ identically, as we needed to ensure.

Concerning possible components of $E_{ab}$ of negative b.w., one has to consider separately those of b.w. $-1$ (which may admit $s=\pm1$) and $-2$ (with $s=0$). The former can only consist of terms linear in $\cF$, its covariant derivatives and their complex conjugates (appropriately contracted with the metric and/or the Weyl tensor -- only the b.w. 0 part of the latter giving a non-zero contribution). The latter may be constructed out of terms quadratic in $\cF$ and its covariant derivatives (namely, products of two terms of b.w. $-1$)\footnote{To be precise, by a term quadratic in $\cF$ we actually mean a product $\cF\bar\cF$, which indeed has $(-2,0)$ if $\cF$ has $(-1,-1)$ (and similarly for terms quadratic in the covariant derivatives).} or terms linear in $\BS$ or linear in the covariant derivatives of $\cF$ or of $\BR$ (both in general contain also terms of b.w. $-2$). We note that terms linear in the Weyl tensor cannot contribute here, since those with $b=-2$ have necessarily $s=\pm2\neq0$. 

Let us first show that components of $E_{ab}$  of b.w. $-1$ are identically zero. These consist of terms linear in (covariant derivatives of) $\cF$ contracted with arbitrary powers of the b.w. 0 part of the curvature tensor, resulting in a term with weights $(-1,-1)$ which is, in addition, invariant under the Sachs symmetry \cite{GHP,penrosebook1} $\bl\mapsto\bbm$, $\bbm\mapsto-\bl$, $\bar\bbm\mapsto\bn$, $\bn\mapsto-\bar\bbm$  (as follows from the discussion in appendix~\ref{app_appl_null}). The only 2-tensor with such properties is (up to an overall factor) $\cF$ itself, which is antisymmetric and cannot thus contribute to components of $E_{ab}$ of b.w. $-1$, as we wanted to prove. A similar discussion applies to terms with weights $(-1,+1)$ constructed out of $\bar\cF$ (cf. again appendix~\ref{app_appl_null} and, in particular, footnote~\ref{footn_Sachs}).

As for components of $E_{ab}$ of b.w. $-2$, they must necessarily have $s=0$, since they can only be terms proportional to $\ell_a\ell_b$. It follows from the discussion in appendix~\ref{app_appl_null} that terms linear in the covariant derivatives of $\cF$ cannot contribute, since those do not possess components with $(b,s)=(-2,0$). Terms quadratic in $\cF$ and its covariant derivatives do instead contribute, as we now detail. Obviously the only possible term quadratic in $\cF$ is $\mathcal{F}_{ac}\mathcal{\bar F}_b^{\ c}=4\Phi_2\bar\Phi_2\ell_a \ell_b$. When considering terms quadratic in $\nabla^k\cF$ it is useful to first note that\footnote{Cf. footnote~\ref{foot_GHP} and appendix~\ref{subsec_GHP} for details on the GHP notation. Let us only mention here that in the spacetime~\eqref{rec_metric_spec} one has $\eth'\Phi_2=(P\Phi_2)_{,\bar\zeta}=(P^2\bar f)_{,\bar\zeta}$ (which is necessarily non-zero since $\Lambda_0\neq0\neq f$), and higher-order derivatives can be expressed as  $\eth'^{k+1}\Phi_2=P^{-k}(P^{k+1}\eth'^k\Phi_2)_{,\bar\zeta}$.\label{foot_GHP_2}} 
\be
	\nabla_c\mathcal{F}_{ab}=-\mathcal{F}_{ab}\Phi_2^{-1}(k_c\thorn'-m_c\eth')\Phi_2 .
	\label{der_F_null}
\ee	
From appendix~\ref{app_appl_null} it follows that, in the product of a term linear in $\nabla^k\cF$ with a term linear in $\nabla^h\bar\cF$, products of components with weights, respectively, $(-1,-1)$ and $(-1,+1)$ (possible only if $k$ and $h$ are both even) will be either zero or again proportional (via a power of $\Lambda_0$) to $\Phi_2\bar\Phi_2\ell_a \ell_b$. For example, using~\eqref{der_F_null} one finds
\be
\mathcal{F}_{ac}\Box\mathcal{\bar F}_b^{\ c}=8\Lambda_0\Phi_2\bar\Phi_2\ell_a \ell_b  , \qquad \Box\mathcal{F}_{ac}\Box\mathcal{\bar F}_b^{\ c}=16\Lambda_0^2\Phi_2\bar\Phi_2\ell_a \ell_b .
\ee
However, $\nabla^k\cF$ may also contain components with $(-1,s)$ and $s<-1$, which are of the form ~\eqref{DF-1s} and thus give rise also to components $(-2,0)$ proportional to $\eth'^j\Phi_2\eth^j\bar\Phi_2$ in products $\nabla^k\cF\nabla^h\bar\cF$ (with $0\le j\le\mbox{min}\{k,h\}$ and $j,k,h$ having the same parity). One has, for example (cf. Lemma~D.3 of \cite{KucOrt19}), 
\be
	\nabla_d\mathcal{F}_{ac}\nabla^d\mathcal{\bar F}_b^{\ c}=4\eth'\Phi_2\eth\bar\Phi_2\ell_a \ell_b , 
\ee	
and
\be
	\nabla_d\mathcal{F}_{ac}\Box(\nabla^d\mathcal{\bar F}_b^{\ c})=20\Lambda_0\eth'\Phi_2\eth\bar\Phi_2\ell_a \ell_b  , \qquad  \nabla_e\nabla_d\mathcal{F}_{ac}\nabla^e\nabla^d\mathcal{\bar F}_b^{\ c}=4(\eth'^2\Phi_2\eth^2\bar\Phi_2+\Lambda_0^2\Phi_2\bar\Phi_2)\ell_a \ell_b .
\ee

Finally, theorem~\ref{th_T_null} in appendix~\ref{app_appl_null} implies that terms linear in the traceless Ricci tensor or in the covariant derivatives of the Riemann tensor can only contribute to $E_{ab}$ with a linear combination of terms of the form $\T^k\Phi_{22}$ for $k\ge0$ (or, equivalently, $\eth'^k\eth^{k}\Phi_{22}$). 
For example one finds
\beqn
 & & C_{abde;cf}C^{bfde}=4\Lambda_0^2\ell_a\ell_c\Phi_{22} , \qquad C_{abcd}^{\phantom{abcd};bd}=\ell_a\ell_c\T\Phi_{22} , \nonumber \\
 & & C_{abcd\phantom{;bde}e}^{\phantom{abcd};bde}=\ell_a\ell_c(\T^2+2\Lambda_0\T)\Phi_{22} .
\eeqn

Combining the results obtained above, we have shown that, in the spacetime~\eqref{rec_metric_spec} with \eqref{rec_Phi2}, any rank-2 tensor constructed from $\bF$, $\BR$ and their covariant derivatives of arbitrary order takes the form
 \be
  E_{ab}=\lambda_0 g_{ab}+2\ell_a\ell_b\bigg(\sum_{k=0}^{N_1}a_k\T^k\Phi_{22}+\sum_{j=0}^{N_2}c_j\eth'^j\Phi_2\eth^j\bar\Phi_2\bigg) , 
	\label{T_null_main}
 \ee
where $N_1,N_2\in \mathbb{N}$, $\lambda_0$, $a_k$ and $c_j$ are constants, and $\Phi_{22}=\frac{1}{2}\T H^{(0)}$. This is clearly of the required form~\eqref{univ_E} (recall~\eqref{T_null_gen}), as we wanted to prove.

However, the coefficient $b_1$, given by
\be
  b_1=(\Phi_2\bar\Phi_2)^{-1}\bigg(\sum_{k=0}^{N_1}a_k\T^k\Phi_{22}+\sum_{j=0}^{N_2}c_j\eth'^j\Phi_2\eth^j\bar\Phi_2\bigg) ,
	\label{b1_null}
\ee
will in general be a spacetime function. This means that~\eqref{Einst} will {\em not}, as it stands, be compatible with~\eqref{Einst_gen} -- namely, one should add $2\Phi_2\bar\Phi_2(\kappa b_1-\kappa_0)\ell_a\ell_b$ to the RHS of~\eqref{Einst} (which is precisely the additional pure radiation term already mentioned in sections~\ref{subsec_backgr} and \ref{subsec_4D_null_gen}). This affects only one component of the latter, which thus reads
\be
 \T H^{(0)}=2\kappa b_1\Phi_2\bar\Phi_2 ,
 \label{Laplac_null}
\ee
and can be solved (at least in principle) to determine $H^{(0)}$ and thus to fully characterize the spacetime metric~\eqref{rec_metric_spec}.  
We note that while the result~\eqref{T_null_main}--\eqref{Laplac_null} is theory-independent, once a particular theory has been specified, the corresponding Einstein equation~\eqref{Einst_gen} will dictate the specific form of the tensor $E_{ab}$ and thus the precise values of the constants appearing in~\eqref{T_null_main}--\eqref{Laplac_null}. The obtained solution of~\eqref{Laplac_null} will thus be specific to the considered theory. This should be contrasted with the class of (strongly) universal solutions with $\Lambda_0=0$ obtained in \cite{KucOrt19}, which can be fully specified as solutions of the Einstein-Maxwell theory (i.e., also fixing the function $H^{(0)}$, with no need to include additional pure radiation) and then automatically solve {\em also} higher-order theories (in other words, there one has $E_{ab}=T_{ab}$ identically for any of the possible theories specified in \cite{KucOrt19}).

We finally observe that the solutions~\eqref{rec_metric_spec}, \eqref{rec_Phi2} with~\eqref{Laplac_null} represent in the Einstein-Maxwell theory electromagnetic waves accompanied by aligned gravitational waves and pure radiation in the (anti)-Nariai universe \cite{PodOrt03,Khlebnikov86,Ortaggio02}. On the other hand, they are also solutions of any modified field equations~\eqref{Einst_gen}, \eqref{Maxw_gen}, in which case the pure radiation term is absent. The spacetime is generically of Petrov type~II and becomes of type~D iff $(PH^{(0)})_{,\zeta\zeta}=0$.\footnote{A result of \cite{VandenBergh89} states that Einstein-Maxwell solutions of type D with a null Maxwell field do not exist in the Kundt class, when the only contribution to the energy-momentum tensor comes from the Maxwell field. However, this does not apply here precisely because we have allowed also for pure radiation.} For these solutions, no non-zero invariants can be constructed using $\cF$, while all  curvature invariants are constant (in other words, $\cF$ is VSI and the metric is CSI).

\section{Locally homogeneous spacetimes}

\label{sec_homog}

By definition, a homogeneous spacetime admits a transitive group of motions. It follows from the discussion in section~12.1 of \cite{Stephanibook} (see also, e.g., \cite{petrov}) that in the case of a multiply-transitive group of motions, either: (i) the spacetime is Kundt (after excluding certain metrics not compatible with a Maxwell field); or (ii) there exists a simply-transitive subgroup $G_4$. In case~(i), imposing the Einstein equation~\eqref{Einst} implies that the spacetime is degenerate Kundt (see also footnote~\ref{footn_deg}), and one is thus reduced to the analysis of sections~\ref{sec_4D_degK_nonnull} and \ref{sec_4D_null}. We can therefore restrict ourselves here to the case when the group is (or contains) a simply-transitive $G_4$. Using a complex null tetrad of invariant vectors, the spin coefficients and the curvature components are {\em constant} (cf., e.g., \cite{Ozsvath65a,Ozsvath65b,RyaShebook,Stephanibook}). Therefore, a Lorentz transformation with constant parameters enables one to align the (invariant) frame to the energy-momentum tensor of the Maxwell field, and thus to $\cF$ itself (recall $\Phi_{ij}=\kappa_0\Phi_{i}\bar\Phi_{j}$).

\subsection{Non-null fields}

Using the frame described above, the Maxwell field is given by \eqref{F1_0} and its energy-momentum tensor by~\eqref{T_nonnull}, where $\Phi_1$ is a constant by~\eqref{complex_inv} (this means that $\cF$ shares the symmetries of the metric, i.e., it is {\em inheriting}). From Maxwell's equation \cite{Stephanibook} one readily obtains
\be
 \rho=\pi=\tau=\mu=0 ,
 \label{Maxw_hom}
\ee
using which one can compute\footnote{This is the same tensor computed for the degenerate Kundt case in~\eqref{dFdbF} and, {\em mutatis mutandis}, comments similar to those given in footnote~\ref{footn_Hornd} apply also here.}
\be
	\nabla_d\mathcal{F}_{(a|c}\nabla^c\mathcal{\bar F}_{|b)}^{\ d}=16|\Phi_1|^2\left(|\lambda|^2\ell_a \ell_b+|\sigma|^2 n_an_b+\bar\kappa\nu m_am_b+\kappa\bar\nu\bar m_a\bar m_b\right) .
	\label{dFdbF_hom}
\ee

The above quantity must vanish in order to fulfill the assumption~\eqref{univ_E}, thus giving
\be
 \lambda=\sigma=\bar\kappa\nu=0 .
\ee 
Therefore either $\kappa=0$ or $\nu=0$, meaning (with \eqref{Maxw_hom}) that either $\bl$ or $\bn$ is a Kundt vector field.\footnote{Further analysis reveals that in fact both $\kappa=0$ {\em and} $\nu=0$ (cf. \cite{Ozsvath65b} for related computations) -- this is however not relevant to the present discussion.} This case is thus already contained in section~\ref{sec_4D_degK_nonnull}.

In passing, let us note that Ozsv\'ath \cite{Ozsvath65b} obtained an Einstein-Maxwell solution with a homogeneous metric (of Petrov type~I) admitting a simply-transitive $G_4$ and an inheriting, non-null Maxwell field (recently shown to be the unique such solution \cite{AndTor20}). One can verify that it is {\em not} a universal solution, precisely because the above tensor~\eqref{dFdbF_hom} is non-zero.

\subsection{Null fields}

Using the invariant frame mentioned above, the Maxwell field is given by~\eqref{F_null}. The Maxwell equation (or the Mariot-Robinson theorem \cite{Stephanibook}) then gives
\be
 \kappa=0=\sigma ,
\ee
i.e., $\bl$ is geodesic and shearfree. Thanks to $\Phi_{00}=\Phi_{01}=\Phi_{02}=0$, the Goldberg-Sachs theorem \cite{GolSac62} (cf. also theorem~7.1 of \cite{Stephanibook}) further gives
\be
 \Psi_0=0=\Psi_1 ,
\ee
while the Sachs equation (cf. the NP equation~(7.21a) of \cite{Stephanibook}) reduces to
\be
 0=\rho^2+(\epsilon+\bar\epsilon)\rho .
 \label{Sachs_homog}
\ee

With the previous results and $\Phi_{11}=\Phi_{12}=0$, the Bianchi equation~(7.32k) of \cite{Stephanibook} takes the form
\be
 0=(\rho+\bar\rho-2\epsilon-2\bar\epsilon)\Phi_{22} ,
\ee
which thus requires $\rho+\bar\rho-2\epsilon-2\bar\epsilon=0$. However, compatibility of the latter with~\eqref{Sachs_homog} gives
\be
 \rho =0 ,
\ee
so that $\bl$ is Kundt. This case therefore belongs to the discussion of section~\ref{sec_4D_null}, as far as spacetimes of Petrov type II and D are concerned. Let us only mention here that certain homogeneous Kundt (plane wave) metrics of type N and O in the presence of a null Maxwell fields and $\Lambda_0=0$ (cf. theorem~12.1 of \cite{Stephanibook}) are known to describe universal solutions \cite{KucOrt19}.

\section{Example~I: nonlinear electrodynamics (NLE)}

\label{sec_NLE}

\label{subsec_NLE_gen}

Nonlinear modifications of Maxwell's theory were originally proposed in order to cure the divergent electron's self-energy, most famously by Born and Infeld \cite{Born33,BorInf34}. An overview of more general NLE can be found, e.g., in \cite{Plebanski70}. 
For simplicity, here we mostly restrict ourselves to NLE minimally coupled to Einstein's gravity (but some comments on Einstein-Weyl gravity are also given in section~\ref{subsubsec_EW}). The theory is thus given by
\be
 S=\int\d^4x\sqrt{-g}\left[\frac{1}{\kappa}(R-2\Lambda)+L(I,J)\right] ,
 \label{NLE}
 \ee
where $L$ is a (in principle arbitrary) function of the two algebraic invariants~\eqref{IJ}. In~\eqref{Einst_gen} and \eqref{Maxw_gen} one thus has (see, e.g., \cite{Peres61,Plebanski70})
\beqn
 & & E_{ab}=-2L_{,I}T_{ab}+\frac{1}{2}g_{ab}(L-IL_{,I}-JL_{,J}) , \label{E_NLE} \\
 & & H^{ab}=L_{,I}F^{ab}+L_{,J}{}^{*}\!F^{ab} , \label{M_NLE}
\eeqn
with $T_{ab}$ as in~\eqref{T}.

On a solution of the Maxwell equation, by~\eqref{Maxw} and $\bF=\d\bA$ one has $\nabla_b F^{ab}=0=\nabla_b {}^{*}\!F^{ab}$. Recalling also that for the fields considered in this paper the invariants $I$ and $J$ are constant (cf.~\eqref{complex_inv}), it is obvious that~\eqref{Maxw_gen} is satisfied identically, while the tensor~\eqref{E_NLE} is precisely of the required form~\eqref{univ_E}, with both $b_1$ and $b_2$ (which can be read off from \eqref{E_NLE}) being constants.\footnote{We emphasize this is a very special situation for which $b_1=$const also in the null case (i.e., eq.~\eqref{E_NLE} with \eqref{T_null_gen} means that in~\eqref{T_null_main} one has $a_k=0$ for $k\ge0$ and $c_j=0$ for $j\ge1$) so that $(\bg,\bF)$ can be taken to be an Einstein-Maxwell solution also there -- {\em without} the additional pure radiation that needs to be included in the general discussion of sections~\ref{subsec_backgr} and \ref{subsubsec_rec_spec}.} In other words, the field equations~\eqref{Einst_gen}, \eqref{Maxw_gen} of any theory~\eqref{NLE} are satisfied identically by the pairs $(\bg,\bF)$ identified in sections~\ref{sec_4D_degK_nonnull} and \ref{sec_4D_null} (namely \eqref{4D_nonnull_metric}, \eqref{4D_nonnull_F} and \eqref{rec_metric_spec}, \eqref{rec_Phi2} {\em with \eqref{H0_null}}), provided the algebraic constraints~\eqref{algebr_constr} admit a real solution. Violations of the latter occur, e.g., in special theories such that, for a given solution $(\bg,\bF)$, one of the quantities $L$, $L_{,I}$ or $L_{,J}$ becomes singular, thus giving rise to ill-defined terms in~\eqref{E_NLE} or in~\eqref{M_NLE} -- see section~\ref{subsec_modmax} below for an example. Another exception arises when (on a given solution) $L_{,I}=0$, so that $b_1=0$ in~\eqref{univ_E} (as happens, e.g., for stealth fields \cite{Smolic18}).
In the following sections~\ref{subsec_BI} and \ref{subsec_modmax} the general results just described will be exemplified in the case of two specific theories of NLE that are of particular interest.

It should be observed that, in the case of null fields, it was already known to Schr{\"o}edinger \cite{Schroedinger35,Schroedinger43} that any~\eqref{M_NLE} solves \eqref{Maxw_gen}, while the validity of condition~\eqref{univ_E}  was subsequently pointed out in \cite{Kichenassamy59,KreKic60,Peres61}.\footnote{The articles \cite{Schroedinger35,Schroedinger43,Kichenassamy59,KreKic60,Peres61} focused on a subset of nonlinear theories of electrodynamics such that the ``pathological'' cases mentioned above do not occur (typically by requiring that in the limit of small $I$ and $J$ the linear Maxwell theory is recovered, i.e., $L\approx-I/2$).} Our results extends also to non-null fields and to theories of gravity other than Einstein's (although in this section we have exemplified only the latter).

We further note that the simple structure of the field equations of NLE enables one to easily identify other Einstein-Maxwell solutions that also solve NLE (but are not universal), in addition to those of sections~\ref{sec_4D_degK_nonnull} and \ref{sec_4D_null}. An example with a non-null field and a {\em non-Kundt} metric is given by the homogenous solution obtained by Ozsv\'ath mentioned in section~\ref{sec_homog} \cite{Ozsvath65b,AndTor20}. Examples with a null field are contained, e.g., in the Robinson-Trautman family \cite{Stephanibook}.

\subsection{Born-Infeld theory}

\label{subsec_BI}

The celebrated NLE of Born and Infeld~\cite{BorInf34} is given by
\be
 L(I,J)=2b^2\left(1-L_0\right) , \qquad L_0=\sqrt{1+\frac{I}{2b^2}-\frac{J^2}{16b^4}} ,
\ee
where $b$ is a constant parameter with the dimension of an inverse length (such that Maxwell's theory is recovered for $b\to\infty$). Then~\eqref{E_NLE} takes the form~\eqref{univ_E} with
\be
 b_1=\frac{1}{L_0} , \qquad b_2=b^2\left(1-\frac{1}{L_0}-\frac{I}{4b^2L_0}\right) .
\ee 

\subsubsection{Non-null fields}

For the Einstein-Maxwell solution~\eqref{4D_nonnull_metric}, \eqref{4D_nonnull_F} one has (cf.~\eqref{complex_inv}, \eqref{F1_0}) $I=-4(\Phi_1^2+\bar\Phi_1^2)$ and $J=4i(\Phi_1^2-\bar\Phi_1^2)$, where $\Phi_1$ is a complex constant. Reparametrizing the latter as 
\be
 \Phi_1=\frac{1}{\sqrt{2}}\rho_0e^{i\theta_0/2} ,
 \label{reparam}
\ee
constraints~\eqref{algebr_constr} become
\be
 \Lambda-\Lambda_0=\kappa_0b^2\left[\sqrt{\left(1-\frac{\rho_0^2}{b^2}\cos\theta_0\right)^2-\frac{\rho_0^4}{b^4}}-1+\frac{\rho_0^2}{b^2}\cos\theta_0\right] , \qquad \kappa=\kappa_0\sqrt{\left(1-\frac{\rho_0^2}{b^2}\cos\theta_0\right)^2-\frac{\rho_0^4}{b^4}} , 
\ee
where the parameter $\theta_0$ reflects the consequences of a duality rotation. Type~D and O solutions of this type (i.e., \eqref{4D_nonnull_metric} with $h=0$) in the Born-Infeld electrodynamics were already obtained in \cite{Morales82} (see also \cite{PleMor86} for more general NLE).\footnote{According to~\cite{Morales82}, earlier unpublished results were obtained by Jan Slav\'{\i}k: 
``It should be noticed that Slav\'{\i}k, looking for some solutions in nonlinear electrodynamics (NLE), mentions already that, in the conformally flat subcase, it is possible to assign the BR metric the role of a carrier of a solution to nonlinear dynamics (J. Slav\'{\i}k: Doctoral Thesis, 1976, Institute of Theoretical Physics of the University of Warszawa, Poland)''.}

\subsubsection{Null fields}

When $I=0=J$ one has (on-shell) $L_0=1$, so that $b_1=1$, $b_2=0$ and~\eqref{algebr_constr} reduces to
\be
 \Lambda=\Lambda_0 , \qquad \kappa=\kappa_0 ,
\ee
i.e. any Einstein-Maxwell solution with a null field (and in particular the one given by \eqref{rec_metric_spec}, \eqref{rec_Phi2}, \eqref{H0_null}) also solves Einstein's gravity coupled to the electrodynamics of Born and Infeld, with no need to redefine $\Lambda_0$ and $\kappa_0$. This fact was already known \cite{Kichenassamy59,KreKic60,Peres61}.

\subsection{ModMax theory}

\label{subsec_modmax}

The recently proposed \cite{Bandosetal20} ModMax electrodynamics is of particular interest in that it preservers both $SO(2)$ duality and conformal invariance  (see also \cite{Kosyakov20}). It is described by 
\be
  L(I,J)=-\frac{1}{2}I\cosh\gamma+\frac{1}{2}\sqrt{I^2+J^2}\sinh\gamma ,
	\label{ModMax}
\ee
where $\gamma$ is a dimensionless parameter\footnote{Not to be confused with the NP/GHP coefficient used in section~\ref{sec_4D_null} and appendix~\ref{subsec_GHP}.} (see \cite{Bandosetal20} for physical reasons to restrict to $\gamma\ge0$, with $\gamma=0$ corresponding to Maxwell's theory).
Here~\eqref{E_NLE} takes the form~\eqref{univ_E} with
\be
  b_1=\cosh\gamma-\frac{I}{\sqrt{I^2+J^2}}\sinh\gamma , \qquad b_2=0 .
\ee 

Since here neither $E_{ab}$ nor $H^{ab}$ in eqs.~\eqref{E_NLE}, \eqref{M_NLE} are well-defined for null fields (i.e., for $I=0=J$) \cite{Bandosetal20},\footnote{The field equations for the electromagnetic field, however, remain well-behaved in the Hamiltonian formalism \cite{Bandosetal20}. Thanks to Dmitri Sorokin for useful comments on this point.} we will consider only the non-null case, i.e., the solution~\eqref{4D_nonnull_metric}, \eqref{4D_nonnull_F}. Then~\eqref{algebr_constr} gives
\be
 \Lambda=\Lambda_0 , \qquad \kappa=\frac{\kappa_0}{\cosh\gamma+\sinh\gamma\cos\theta_0} .
\ee
where $\theta_0$ is defined as in~\eqref{reparam}. Here one should exclude special fine-tuned configurations with $\cosh\gamma+\sinh\gamma\cos\theta_0$, which correspond to $b_1=0$ and thus $E_{ab}=0$, i.e., to stealth configurations of the theory~\eqref{ModMax}, for which the spacetime metric is Einstein. 

Solutions of this type (in the case with $k_2=0$ in~\eqref{4D_nonnull_metric}) were considered in \cite{Flores-Alfonsoetal21}.

\subsubsection{Extension to Einstein-Weyl gravity}

\label{subsubsec_EW}

The results given above refer to Einstein's gravity coupled to ModMax electrodynamics (i.e., \eqref{NLE} with \eqref{ModMax}). Given the conformal invariance of the latter, it may be now instructive to consider an example where also gravity is modified by the addition of a conformal invariant term. A well-known theory with such a property is given by Einstein-Weyl gravity, corresponding to the action
\be
 S=\int\d^4x\sqrt{-g}\left[\frac{1}{\kappa}(R-2\Lambda)-\alpha_0C_{abcd}C^{abcd}+L(I,J)\right] ,
 \label{NLE_EW}
 \ee
with \eqref{ModMax}, where $\alpha_0$ is a coupling constant with the dimension of a length squared (Weyl conformal gravity is obtained for $\kappa^{-1}=0$). Here \eqref{M_NLE} is unchanged, while \eqref{E_NLE} becomes \cite{Bach21,Buchdahl53}
\be
 E_{ab}=4\alpha_0B_{ab}-2L_{,I}T_{ab}+\frac{1}{2}g_{ab}(L-IL_{,I}-JL_{,J}) , 
 \label{E_NLE_EW}
\ee
where $B_{ab}$ is the (symmetric, traceless, conserved) Bach tensor
\be
	B_{ab}=\big( \nabla^c \nabla^d + \textstyle{\frac{1}{2}} R^{cd} \big) C_{acbd} . 
	\label{Bach}
\ee

As above, it makes sense here to consider only the non-null field solution~\eqref{4D_nonnull_metric}, \eqref{4D_nonnull_F}. For the latter one easily finds (in agreement with the results of section~\ref{sec_4D_degK_nonnull})
\be
 \nabla^c \nabla^dC_{acbd}=0 , \qquad R^{cd}C_{acbd}=\frac{4}{3} \kappa _0 \Lambda _0	T_{ab}  ,
\ee
which with \eqref{univ_E}, \eqref{algebr_constr}, \eqref{E_NLE_EW}, \eqref{Bach} and \eqref{reparam} gives 
\be
 \Lambda=\Lambda_0 , \qquad \kappa=\frac{\kappa_0}{\cosh\gamma+\sinh\gamma\cos\theta_0+\frac{8}{3}\alpha_0 \kappa _0 \Lambda _0} .
\ee
Similarly as in the case of Einstein gravity discussed above, special configurations with $\cosh\gamma+\sinh\gamma\cos\theta_0+\frac{8}{3}\alpha_0 \kappa _0 \Lambda _0=0$ describe stealth fields (i.e., $E_{ab}=0$) and should be considered separately.

\section{Example~II: Horndeski's electrodynamics}

\label{sec_Horn}

Horndeski \cite{Horndeski76} obtained the unique theory (constructed from a Lagrangian depending on the metric, the vector potential and their derivatives) such that: (i) the corresponding field equations are of second order; (ii) in the presence of sources they are compatible with charge conservation; (iii) the equation for the electromagnetic field reduces to Maxwell's equation in flat space. The theory of~\cite{Horndeski76} is given by 
\be
 S=\int\d^4x\sqrt{-g}\left[\frac{1}{\kappa}(R-2\Lambda)-\beta_0 F_{ab}F^{ab}-\gamma_0 F_{ab}F^{cd}\,{}^{*}\!R^{*ab}_{\phantom{*ab}cd}\right] ,
 \label{Horn}
 \ee
where $\beta_0$ and $\gamma_0$ are coupling constants (dimensionless and with the dimension of a length squared, respectively), and the standard Einstein-Maxwell theory is recovered for $\gamma_0=0$. In~\eqref{Einst_gen} and \eqref{Maxw_gen} this gives rise to \cite{Horndeski76,Horndeski76_unp,HorWai77}
\beqn
 & & E_{ab}=2\beta_0 T_{ab}+2\gamma_0\left(F^{ce}F^{d}_{\phantom{d}e}\,{}^{*}\!R^*_{acbd}+\nabla_d{}^{*}\!F_{ac}\,\nabla^c{}^{*}\!F_{b}^{\phantom{b}d}\right) , \label{E_Horn} \\
 & & H^{ab}=\beta_0 F^{ab}+\gamma_0 F_{df}\,{}^{*}\!R^{*abdf} , \label{M_Horn}
\eeqn
with $T_{ab}$ as in~\eqref{T}.

\subsection{Non-null fields}

For the Einstein-Maxwell solution~\eqref{4D_nonnull_metric}, \eqref{4D_nonnull_F}, the term $F_{df}\,{}^{*}\!R^{*abdf}$ in~\eqref{M_Horn} becomes a linear combination with constant coefficients of $F^{ab}$ and ${}^{*}\!F^{ab}$, so that~\eqref{Maxw_gen} is satisfied identically, in agreement with the results of section~\ref{sec_4D_degK_nonnull}. In~\eqref{E_Horn} one has $\nabla_d{}^{*}\!F_{ac}\,\nabla^c{}^{*}\!F_{b}^{\phantom{b}d}=0$ (see section~\ref{sec_4D_degK_nonnull}), while one can compute
\be
 F^{ce}F^{d}_{\phantom{d}e}\,{}^{*}\!R^*_{acbd}=-\left[\Lambda_0+\kappa_0(\Phi_1^2+\bar\Phi_1^2)\right]T_{ab}+\left[\Lambda_0(\Phi_1^2+\bar\Phi_1^2)+4\kappa_0\Phi_1^2\bar\Phi_1^2\right]g_{ab} .
 \label{FFSRS}
\ee

Using~\eqref{E_Horn}, \eqref{FFSRS}, \eqref{univ_E} and again the parametrization~\eqref{reparam}, constraints~\eqref{algebr_constr} become
\be
 \Lambda-\Lambda_0=\frac{\kappa_0\gamma_0\rho_0^2(\Lambda_0\cos\theta_0+\kappa_0\rho_0^2)}{\beta_0-\gamma_0(\Lambda_0+\kappa_0\rho_0^2\cos\theta_0)} , \qquad \kappa=\frac{\kappa_0}{2\beta_0-2\gamma_0(\Lambda_0+\kappa_0\rho_0^2\cos\theta_0)} .
\ee

Here one should exclude special fine-tuned configurations such that $\beta_0-\gamma_0(\Lambda_0+\kappa_0\rho_0^2\cos\theta_0)=0$, which correspond to $b_1=0$ and thus to stealth configurations in the theory~\eqref{Horn}, for which the spacetime metric is Einstein (more precisely, these configurations are ``almost stealth'', since in general $b_2\neq0$).

Solutions of the form~\eqref{4D_nonnull_metric}, \eqref{4D_nonnull_F} were also constructed in \cite{Horndeski78_Birk} in the case $h=0$ with $k_2>0$.

\subsection{Null fields}

\label{subsec_Horn_null}

For the fields $(\bg,\bF)$ given by~\eqref{rec_metric_spec}, \eqref{rec_Phi2}, in~\eqref{M_Horn} one finds $F_{df}\,{}^{*}\!R^{*abdf}=0$, therefore~\eqref{Maxw_gen} is satisfied trivially, in agreement with the results of section~\ref{sec_4D_null}. The terms in~\eqref{E_Horn} take the form (recall~\eqref{der_F_null} and footnotes~\ref{footn_Hornd}, \ref{foot_GHP}, and \ref{foot_GHP_2})
\beqn
 & & F^{ce}F^{d}_{\phantom{d}e}\,{}^{*}\!R^*_{acbd}=-\Lambda_0 T_{ab} , \label{Horn_null_corr1} \\ 
 & & \nabla_d{}^{*}\!F_{ac}\,\nabla^c{}^{*}\!F_{b}^{\phantom{b}d}=2\eth'\Phi_2\eth\bar\Phi_2\ell_a \ell_b , \label{Horn_null_corr2} 
\eeqn

Hence~\eqref{E_Horn}, \eqref{univ_E} give $b_2=0$ and thus, by~\eqref{algebr_constr},
\be
 \Lambda=\Lambda_0 .
\ee
However, $b_1$ as determined from~\eqref{E_Horn} (with~\eqref{T_null_gen}, \eqref{Horn_null_corr1}, \eqref{Horn_null_corr2}) is not a constant, so that the field equation~\eqref{Laplac_null} (i.e., the only remaining component of~\eqref{Einst_gen} to be solved) reduces to the following partial differential equation
\be
 \T H^{(0)}=4\kappa\left[(\beta_0-\gamma_0\Lambda_0)\Phi_2\bar\Phi_2+\gamma_0\eth'\Phi_2\eth\bar\Phi_2\right] .
 \label{H0_Horn_null}
\ee
This can be solved (at least in principle) to determine the metric function $H^{(0)}$ of~\eqref{rec_metric_spec}. We emphasize once again that here $H^{(0)}$ is {\em not} of the form~\eqref{H0_null} as in the electrovac Einstein-Maxwell theory -- i.e., a solution of~\eqref{H0_Horn_null} corresponds on the Einstein-Maxwell side to a null electromagnetic field accompanied by aligned pure radiation, as discussed in more generality in sections~\ref{subsec_backgr} and \ref{subsubsec_rec_spec}. In the limit $\Lambda_0=0$ one recovers the \pp waves obtained in \cite{GurHal78,Horndeski79}.

\section*{Acknowledgments}

I am grateful Sigbj\o rn Hervik for useful comments. This work has been supported by research plan RVO: 67985840 and research grant GA\v CR 19-09659S.

\renewcommand{\thesection}{\Alph{section}}
\setcounter{section}{0}

\renewcommand{\theequation}{{\thesection}\arabic{equation}}

\section{Basics of the GHP formalism}
\setcounter{equation}{0}

\label{subsec_GHP}

In this appendix we summarize the basic equations of the GHP formalism needed for the purposes of the present paper. Some familiarity with the formalism will, however, be assumed -- see \cite{GHP,penrosebook1} for more details.

\subsection{Preliminaries}

Let us introduce a normalized spinor dyad $(o^A,\iota^A)$, such that 
\be
 o_A\iota^A=1 ,
 \label{spinors_norm}
\ee
out of which one can define a standard null tetrad \cite{GHP,penrosebook1}
\be
  \ell^a=o^A\bar o^{A'} , \quad n^a=\iota^A\bar \iota^{A'} , \quad m^a=o^A\bar\iota^{A'} , \quad \bar m^a=\iota^A\bar o^{A'} ,
	\label{basis_spinors}
\ee
as used in the main text of the present paper.\footnote{We warn the reader that in the appendices we use for spinors the more traditional signature convention of \cite{GHP,penrosebook1}, i.e., the basis~\eqref{basis_spinors} with \eqref{spinors_norm} corresponds to {\em minus} the metric~\eqref{g_tetrad} used in the main body of the paper. Furthermore, to make the contact with the NP formalism easier, in the GHP notation we make the standard substitutions \cite{GHP,penrosebook1,Stephanibook} $\tau'=-\pi$, $\sigma'=-\lambda$, $\rho'=-\mu$, $\kappa'=-\nu$, $\gamma'=-\epsilon$, $\beta'=-\alpha$.} In this basis one has
\be
 \varepsilon_{AB}=o_A\iota_B-\iota_Ao_B .
\ee

A boost-spin transformation of the dyad is defined by
\be
 o^A\mapsto \chi o^A , \qquad \iota^A\mapsto \chi^{-1}\iota^A ,
 \label{dyad_transf}
\ee
where $\chi$ is a complex scalar field.

A quantity (a scalar, a spinor or a tensor) $\eta$ is called a {\em weighted quantity} of type $\{p,q\}$ if under~\eqref{dyad_transf} it transforms as 
\[
	\eta\mapsto \chi^p \bar{\chi}^q \eta .
\]
In particular, $o^A$ and $\iota^A$ can be regarded  themselves as spinors of type $\{1,0\}$ and $\{-1,0\}$, respectively. Recall also that complex conjugation interchanges the values of $p$ and $q$.

Equivalently, one can also say that $\eta$ possesses boost and spin weights $(b,s)$ given by
\be
	b=\frac{1}{2}(p+q) , \qquad s=\frac{1}{2}(p-q) , 
	\label{bs}
\ee
which will be useful in the following. Under complex conjugation $b$ is thus invariant whereas $s$ changes sign.

The GHP derivative operators are defined by their action on a quantity of type $\{p,q\}$, namely
\beqn
	& & \thorn\eta=(D-p\epsilon-q\bar\epsilon)\eta  , \qquad \thorn'\eta=(D'-p\gamma-q\bar\gamma)\eta , \\
	& & \eth\eta=(\delta - p\beta -q\bar\alpha)\eta , \qquad \eth'\eta=(\delta'-p\alpha-q\bar\beta)\eta ,
\eeqn
where $(D,D'=\bigtriangleup,\delta,\delta'=\bar\delta)$ are the standard directional derivatives of the Newman-Penrose formalism \cite{NP,GHP,Stephanibook,penrosebook1}.

For later use it is useful to display how the derivative operators act on the basis spinors, i.e., 
\beqn
 & & \thorn o^A=-\kappa\iota^A , \quad \thorn \bar o^{A'}=-\bar\kappa\bar\iota^{A'} , \quad \thorn \iota^A=\pi o^A , \quad \thorn \bar \iota^{A'}=\bar\pi \bar o^{A'} , \label{der_thorn} \\
 & & \thorn' o^A=-\tau\iota^A , \quad \thorn' \bar o^{A'}=-\bar\tau\bar\iota^{A'} , \quad \thorn' \iota^A=\nu o^A , \quad \thorn' \bar \iota^{A'}=\bar\nu\bar o^{A'} , \\
 & & \eth o^A=-\sigma\iota^A , \quad \eth\bar o^{A'}=-\bar\rho\bar\iota^{A'} , \quad \eth \iota^A=\mu o^A , \quad \eth\bar \iota^{A'}=\bar\lambda\bar o^{A'} , \\
 & & \eth' o^A=-\rho\iota^A , \quad \eth'\bar o^{A'}=-\bar\sigma\bar\iota^{A'} , \quad \eth' \iota^A=\lambda o^A , \quad \eth'\bar \iota^{A'}=\bar\mu\bar o^{A'} , \label{der_eth'}
\eeqn
the RHS of each of the above equations defining one of the weighted spin coefficients.

The types of the weighted spin coefficients, of the Maxwell, traceless Ricci and Weyl spinor components and of the derivative operators are given by 
\beqn
& & \kappa : \{3, 1\}, \qquad \sigma : \{3, -1\}, \qquad \rho : \{1, 1\} ,  \qquad \tau : \{1, -1\},\nonumber \\
& & \nu : \{-3, -1\},  \qquad \lambda : \{-3, 1\}, \qquad \mu : \{-1, -1\},  \qquad \pi : \{-1, 1\}, \nonumber \\
& & \Phi_0 : \{2, 0\},  \qquad \Phi_1 : \{0, 0\},  \qquad \Phi_2 : \{-2, 0\}, \nonumber \\
& & \Phi_{00}=\bar\Phi_{00} : \{2, 2\},  \qquad \Phi_{01}=\bar\Phi_{10} : \{2, 0\},  \qquad \Phi_{02}=\bar\Phi_{20} : \{2, -2\}, \label{GHP_weights} \\
& & \Phi_{11}=\bar\Phi_{11} : \{0, 0\},  \qquad \Phi_{12}=\bar\Phi_{21} : \{0, -2\},  \qquad \Phi_{22}=\bar\Phi_{22} : \{-2, -2\}, \nonumber \\
& & \Psi_0 : \{4, 0\},  \qquad \Psi_1 : \{2, 0\},  \qquad \Psi_2 : \{0, 0\}, \qquad  \Psi_3: \{-2,0\},  \qquad \Psi_4 : \{-4, 0\} ,\nonumber \\
& & \thorn : \{1,1\} , \qquad  \thorn' : \{-1,-1\},   \qquad \eth : \{1, -1\},  \qquad \eth' : \{-1,1\} . \nonumber
\eeqn

\subsection{Type~II Kundt spacetimes with an aligned electromagnetic field}

From now on we assume that the spacetime is Kundt of Riemann type II, the electromagnetic field is aligned, and the basis dyad is parallelly transported, i.e.,
\be
  \Psi_0=\Psi_1=0=\Phi_{00}=\Phi_{01} , \qquad \Phi_0=0 , \qquad\kappa=\rho=\sigma=\epsilon=\pi=0 .
	\label{degK_app}
\ee

The commutators thus take the simplified form
\beqn
& & \left[ \thorn,\thorn' \right]=\bar\tau\eth +\tau\eth'-p(\Psi_2+\Phi_{11}-R/24)-q(\bar\Psi_2+\Phi_{11}-R/24) , \label{comm_thth'} \\
& & \left[ \eth,\eth' \right]=(\mu-\bar\mu)\thorn+p(\Psi_2-\Phi_{11}-R/24)-q(\bar\Psi_2-\Phi_{11}-R/24) , \label{comm_etheth'} \\ 
& & \left[ \thorn,\eth \right] =0 , \label{comm_theth} \\  
& & \left[ \thorn',\eth \right] =-\mu\eth-\bar\lambda\eth'-\tau\thorn'+\bar\nu\thorn-p(\bar\tau\bar\lambda+\bar\Psi_3)-q(\mu\tau+\Phi_{12}) , \label{comm_th'eth}  
\eeqn
along with the complex conjugates of the last two equations.

If one also assumes that the energy-momentum tensor originates from the electromagnetic field, i.e., $\Phi_{\alpha\beta}=\kappa_0\Phi_\alpha\bar\Phi_\beta$ (with $\alpha,\beta=0,1,2$), then one also has
\be
 \Phi_{02}=0 , \qquad R=4\Lambda_0 .
\label{degK_app_aligned}
\ee

\section{Proof for non-null fields (section~\ref{sec_4D_degK_nonnull})}
\setcounter{equation}{0}

\label{app_nonnull}

In addition to \eqref{degK_app} and \eqref{degK_app_aligned}, for the solutions of section~\ref{sec_4D_degK_nonnull} we have
\be
 \tau=\mu=\lambda=0 , \qquad \Psi_2=-\frac{\Lambda_0}{3} , \qquad \Psi_3=0 , \qquad \Phi_1=\mbox{const} \qquad \Phi_2=0 , \label{nonnull_cond1} 
\ee 
so that also $\Phi_{12}=\Phi_{22}=0$, and the spin-coefficient equations and Bianchi identities needed in the following read\footnote{The Maxwell equation is satisfied identically here as a consequence of~\eqref{nonnull_cond1} and \eqref{degK_app}.}
\beqn
 & & \thorn\nu=0 , \qquad \eth\nu=0 , \qquad \eth'\nu=\Psi_4 , \label{nonnull_cond2} \\
 & & \thorn\Psi_4 =0 , \qquad \eth\Psi_4=-\nu(3\Psi_2-2\Phi_{11}) . \label{nonnull_cond3}
\eeqn
The only non-zero derivatives of basis spinors \eqref{der_thorn}--\eqref{der_eth'} are given by
\be
 \thorn' \iota^A=\nu o^A , \qquad \thorn' \bar \iota^{A'}=\bar\nu\bar o^{A'} .
 \label{nonnull_spin_der}
\ee

\subsection{1-balanced $s$-balanced spinors}

Let us consider a spinor (or tensor)\footnote{Recall that any tensor field may be interpreted as a spinor field \cite{GHP,penrosebook1}. This will be understood in what follows.} field $\bS$. By a refinement of the notion of  1-balanced tensors introduced in \cite{HerPraPra14} (see also \cite{Pravdaetal02,Coleyetal04vsi,Herviketal15,OrtPra16,HerPraPra17,HerOrtPra18}) it is useful to give the following
\begin{definition}[1-balanced $s$-balanced spinors]
\label{def_1bal_sbal}
A scalar $\eta$ of boost and spin weights ($b,s$) is a ``1-balanced $s$-balanced scalar'' if it satisfies the following two conditions
	\beqn
		& & \eta=0 \quad \mbox{for } b\ge-1 , \qquad \thorn^{-b-1}\eta=0 \quad \mbox{for } b<-1 , \label{1balanced} \\
		& & \thorn^{-b-2}\eth^{-s}\eta=0 \quad \mbox{for } s<0 , \qquad \thorn^{-b-2}\eta=0 \quad \mbox{for } s\ge0 . \label{s_balanced}
	\eeqn
A spinor $\bS$ whose components are all 1-balanced $s$-balanced scalars is a ``1-balanced $s$-balanced spinor''.
\end{definition}

Condition~\eqref{1balanced} defines 1-balanced scalars \cite{HerPraPra14}, therefore a 1-balanced $s$-balanced spinor is, in particular, a 1-balanced spinor. We will consider only 1-balanced $s$-balanced spinors with an integer $b$, which means that $b\le-2$ for all non-zero components of $\bS$. Furthermore, by~\eqref{s_balanced}, components with $b=-2$ can only have $s<0$ (with $\eth^{-s}\eta=0$).

We will need the following
\begin{theorem}[Derivatives of 1-balanced $s$-balanced spinors]
 \label{th_der_1bal_sbal}
 In the spacetime~\eqref{4D_nonnull_metric} the covariant derivative of a 1-balanced $s$-balanced spinor $\bS$ is again a 1-balanced $s$-balanced spinor.
 \end{theorem} 
\begin{proof}
 By assumption, the components of $\bS$ satisfy \eqref{1balanced} and \eqref{s_balanced}. Since the spacetime is degenerate Kundt, it follows from \cite{HerPraPra14,HerOrtPra18} that $\nabla\bS$ is automatically 1-balanced, i.e., eq.~\eqref{1balanced} is satisfied also by the components of $\nabla\bS$. It remains to be shown that the same is true for eq.~\eqref{s_balanced}. To this end, we note that the possible non-zero components of $\nabla\bS$ and their respective boost and spin weights are given in terms of those of $\bS$ (the latter being indicated generically as $\eta$) by 
\beqn
 & & \thorn\eta: (b+1,s) , \qquad \thorn'\eta: (b-1,s) , \qquad \eth\eta: (b,s+1) , \qquad \eth'\eta: (b,s-1) , \label{covder_compts} \\
 & & \nu\eta: (b-2,s-1) , \qquad \bar\nu\eta: (b-2,s+1) , \label{covder_compts2} 
\eeqn
where we have used \eqref{GHP_weights} , \eqref{nonnull_spin_der} and \eqref{bs}, and the components of $\nabla\bS$ are meant up to numerical factors.
Recalling \eqref{1balanced}, \eqref{s_balanced} and the fact that $b\le-2$ for components of $\bS$ (being 1-balanced), it is easy to see that among the components~\eqref{covder_compts}, \eqref{covder_compts2}, those of b.w. $-2$ (if any) can only have s.w.$<0$. Using \eqref{comm_thth'}--\eqref{comm_th'eth}, \eqref{nonnull_cond1}--\eqref{nonnull_cond3} and \eqref{1balanced}, \eqref{s_balanced}, one can also see that the property \eqref{s_balanced} holds also when $\eta$ is replaced by any of the components \eqref{covder_compts}, \eqref{covder_compts2} (and their respective weights), as we wanted to prove.

\end{proof}

Iteratively, it is obvious that the same property holds for any covariant derivative of $\bS$ of arbitrary order, i.e., any components of such covariant derivatives necessarily have b.w. not greater than $-2$, and those of b.w. $-2$ can only have negative s.w..

\subsection{Application to section~\ref{subsec_nonnull_suff}}

\label{app_appl_nonnull}

Let us now employ the above general results for the purposes of the present paper. As discussed in section~\ref{sec_4D_degK_nonnull} (and references therein), the covariant derivative of the self-dual Maxwell tensor $\cF$, of the energy-momentum tensor and of the Weyl tensor are all 1-balanced tensors. It is also easy to see that the covariant derivatives of both the self-dual Maxwell tensor and the self-dual part of the Weyl tensor (or spinor) additionally satisfy condition~\eqref{s_balanced} and are therefore 1-balanced $s$-balanced tensors (spinors) -- by theorem~\ref{th_der_1bal_sbal} this will be true also for their covariant derivatives of any order. The covariant derivative of the energy-momentum tensor~\eqref{der_F_T} is a sum of two terms, the first of which is again a 1-balanced $s$-balanced tensor. The second term in $\nabla_c T_{ab}$ is simply the complex conjugate of the first one, and instead of \eqref{s_balanced} it satisfies its complex conjugate version, namely $\thorn^{-b-2}\eth'^{s}\eta=0$ for $s>0$, and $\thorn^{-b-2}\eta=0$ for $s\le0$ (with $b=-2$, $s=1$). Theorem~\ref{th_der_1bal_sbal} ensures that the first covariant derivative of the first term of $\nabla_c T_{ab}$ is also a 1-balanced $s$-balanced tensor, while an equivalent (up to complex conjugation) result for the second term is obvious. Iteratively, one can argue similarly for covariant derivatives of $T_{ab}$ of any order.

As a conclusion, terms of b.w.~$-2$ (i.e., those of highest b.w.) in the covariant derivatives of arbitrary order of the (anti-)self-dual Maxwell tensor, the energy-momentum tensor or the (anti-)self-dual Weyl tensor can only have either $s<0$ or $s>0$, but not $s=0$ -- which is indeed the result used in section~\ref{subsec_nonnull_suff}.

\section{Proof for null fields with $D\Psi_4=0$ (section~\ref{subsubsec_rec_spec})}
\setcounter{equation}{0}

\label{app_null}

In addition to \eqref{degK_app} and \eqref{degK_app_aligned}, for the solutions of section~\ref{subsubsec_rec_spec} we have
\be
 \tau=\mu=\lambda=0 , \qquad \Psi_2=-\frac{\Lambda_0}{3} , \qquad \Psi_3=0 , \qquad \Phi_1=0 , \label{null_cond1} 
\ee 
so that also $\Phi_{11}=\Phi_{12}=0$, and the spin-coefficient equations, the Bianchi identities and the Maxwell equations needed in the following read
\beqn
 & & \thorn\nu=0 , \qquad \eth\nu=\Phi_{22} , \qquad \eth'\nu=\Psi_4 , \label{null_cond2} \\
 & & \thorn\Psi_4 =0 , \qquad \eth\Psi_4-\eth'\Phi_{22}=-3\nu\Psi_2 , \qquad \thorn\Phi_{22}=0 , \label{null_cond3} \\
 & & \thorn\Phi_2 =0 , \qquad \eth\Phi_2 =0 . \label{null_cond4} 
\eeqn
The only non-zero derivatives of basis spinors \eqref{der_thorn}--\eqref{der_eth'} are as in \eqref{nonnull_spin_der}.

Before proceeding, it is useful to recall that the spinor equivalent of the null field~\eqref{F_null} reads \cite{GHP,penrosebook1,Stephanibook}
\be
 \phi_{AB}=\Phi_2o_Ao_B .
 \label{F_null_spinor}
\ee

As discussed in sections~\ref{subsec_backgr} and \ref{sec_4D_null}, we do {\em not} impose the last component of the Einstein equation in the Einstein-Maxwell theory $\Phi_{22}=\kappa_0\Phi_2\bar\Phi_2$.

\subsection{Balanced $s$-balanced spinors}

\label{subsubsec_bal_s_bal}

We can now conveniently adapt the approach of section~\ref{app_nonnull} and refine the notion of balanced spinors given  in \cite{Pravdaetal02} (cf. also \cite{Coleyetal04vsi,OrtPra16,HerPraPra17,HerOrtPra18}) by defining 
\begin{definition}[Balanced $s$-balanced spinors]
\label{def_bal_sbal}
A scalar $\eta$ of boost and spin weights ($b,s$) is a ``balanced $s$-balanced scalar'' if it satisfies the following two conditions
\beqn
 & & \eta=0 \quad \mbox{for } b\ge0 , \qquad \thorn^{-b}\eta=0 \quad \mbox{for } b<0 , \label{balanced} \\ 
 & & \thorn^{-b-1}\eth^{-s}\eta=0 \quad \mbox{for } s<0 , \qquad \thorn^{-b-1}\eta=0 \quad \mbox{for } s\ge0 . 
	\label{b_s_balanced}
\eeqn
A spinor $\bS$ whose components are all balanced $s$-balanced scalars is a ``balanced $s$-balanced spinor''.
\end{definition} 

Condition~\eqref{balanced} defines balanced scalars \cite{Pravdaetal02,Coleyetal04vsi}, therefore a balanced $s$-balanced spinor is, in particular, a balanced spinor.
We will consider only spinors with an integer $b$, which means that $b\le-1$ for all non-zero components of $\bS$. The above definition also means, in particular, that components of $\bS$ with $b=-1$ can only have $s<0$ (with $\eth^{-s}\eta=0$). Recalling also definition~\ref{def_1bal_sbal}, we further note that one has the following series of implications for a spinor: $\bS$ is 1-balanced $s$-balanced $\Rightarrow$ $\bS$ is balanced $s$-balanced $\Rightarrow$ $\bS$ is balanced.

One has the following useful
\begin{theorem}[Derivatives of balanced $s$-balanced spinors]
 \label{th_der_bal_sbal}
	In the spacetime~\eqref{rec_metric_spec} the covariant derivative of a balanced $s$-balanced spinor $\bS$ is again a balanced $s$-balanced spinor.
 \end{theorem} 
\begin{proof}
Since the spacetime is degenerate Kundt, $\nabla\bS$ is a balance spinor as well \cite{Pravdaetal02,Coleyetal04vsi,HerPraPra14,OrtPra16}, i.e., \eqref{balanced} holds true for any of its components, so that it remains to prove that \eqref{b_s_balanced} is also obeyed by all components of $\nabla\bS$. 

The possible non-zero components of the covariant derivative of $\bS$ and their respective boost and spin weights are again given (up to numerical factors) by \eqref{covder_compts}, \eqref{covder_compts2}. Those of b.w. $-1$ (if any) can only have a negative s.w., as can be easily seen using \eqref{balanced}, \eqref{b_s_balanced} and the fact that $b\le-1$ for components of $\bS$ (being balanced). Furthermore, 
using \eqref{comm_thth'}--\eqref{comm_th'eth}, \eqref{null_cond1}--\eqref{null_cond4}  and \eqref{balanced}, \eqref{b_s_balanced}, it follows that the property \eqref{b_s_balanced} holds also when $\eta$ is replaced by any of the components \eqref{covder_compts}, \eqref{covder_compts2} (and their respective weights). This means that the covariant derivative of $\bS$ is also a balanced $s$-balanced spinor, as we wanted to prove. 

\end{proof}

Clearly the above theorem can be applied iteratively, such that covariant derivative of arbitrary order of a balanced $s$-balanced spinor $\bS$ are also balanced $s$-balanced spinors. It follows, in particular, that any components of such covariant derivatives with b.w. $-1$ can only have a negative s.w..

Let us now consider a spinor $\bS$ possessing components that, in addition to \eqref{balanced} and \eqref{b_s_balanced}, further obey
\be
 \thorn^{-b-2}\eth^{-s}\eta=0 \quad \mbox{for } b\le-2, s<0 , \qquad \thorn^{-b-2}\eta=0 \quad \mbox{for } b\le-2, s\ge0 . 
	\label{b_s_extra}
\ee
(Note that \eqref{b_s_extra} is stronger than \eqref{b_s_balanced} but only applies to components with $b\le-2$.) This implies, in particular, that components of $\bS$ with $b=-2$ can only have $s<0$, with $\eth^{-s}\eta=0$.  Similarly as above, it is easy to see iteratively that this property is inherited by covariant derivatives of $\bS$ of any order. It follows that, for all of those, terms with $(b,s)=(-2,0)$ vanish identically.

\subsection{Application to section~\ref{subsubsec_null_suff}}

\label{app_appl_null}

\subsubsection{Components of $\nabla^k\bF$}

The null electromagnetic field \eqref{F_null} (or its spinor equivalent~\eqref{F_null_spinor}) has $(b,s)=(-1,-1)$. Its only component $\Phi_2$ satifies \eqref{null_cond4}, which means that $\cF$ is balanced $s$-balanced (as defined by \eqref{balanced}, \eqref{b_s_balanced}). By theorem~\ref{th_der_bal_sbal}, covariant derivatives of $\cF$ of arbitrary order are also balanced $s$-balanced, which implies that their components with $b=-1$ necessarily have $s<0$. For the complex conjugate $\bar\cF$ and its covariant derivatives, a similar argument shows that components with $b=-1$ necessarily have $s>0$. Therefore, in neither case are components with $b=-1$ and $s=0$ possible. Notice that this conclusion also applies to arbitrary contractions of covariant derivatives $\cF$ or $\bar\cF$ with the metric or with the Weyl tensor. This result is used in section~\ref{subsubsec_null_suff} to prove that $\nabla_b{H}^{ab}=0$ identically, in the context of the Einstein-Maxwell solutions considered there.

Let us now discuss the form of components with $(b,s)=(-1,-1)$ contained in the covariant derivatives (of arbitrary order) of $\cF$. First, let us notice that after each differentiation, the produced components which are proportional to $\nu$ or $\bar\nu$ (cf.~\eqref{covder_compts2}) will always have $b\le-3$, which follows from $\cF$ being balanced and from the first of \eqref{null_cond2}.\footnote{More precisely, such components are {\em 2-balanced}, as defined in \cite{Kuchynkaetal19} (see also Lemma~1 therein).} For the purposes of our discussion these components can thus be neglected, i.e., when considering derivatives $\nabla^k\phi_{AB}$ it suffices to focus on components of the form $\nabla^k\Phi_2$. It follows that components of $\nabla^k\phi_{AB}$ with $(b,s)=(-1,-1)$ can only be produced by applying on $\Phi_2$ an equal number of $\thorn$ and $\thorn'$ operators, and/or an equal number of $\eth$ and $\eth'$ (in any order) -- in particular this implies that they can appear only in derivatives of even order. Now, because of \eqref{null_cond4}, it suffices to consider components of the form $(\ldots)\thorn'^j\Phi_2$ and $(\ldots)\eth'^k\Phi_2$, where the dots represent an arbitrary sequence of derivative operators such that (because of what we have just remarked) the total number of operators $\thorn$ acting on $\Phi_2$ equals that of $\thorn'$, and similarly for $\eth$ and $\eth'$. Thanks to~\eqref{comm_thth'} and \eqref{comm_theth}, all operators $\thorn$ can be shifted to the right until they act on $\Phi_2$ -- up to bearing in mind that every time that $\thorn$ is swapped with $\thorn'$, also an additional term of lower order of differentiation will be produced. Therefore, terms containing $\thorn$ can be iteratively reduced to either zero or to terms of the form $(\ldots)\thorn\thorn'\Phi_2$, which can be further simplified by noticing that
\be
 \thorn\thorn'\Phi_2=-\Lambda_0\Phi_2  ,
 \label{thotho'Phi_2}
\ee  
which follows from~\eqref{comm_thth'}. This can be repeated until all operators $\thorn$ and $\thorn'$ have disappeared. 
There thus remain to consider terms which contain only an equal number of operators $\eth$ and $\eth'$, but no operators $\thorn$ and $\thorn'$. Here one can use~\eqref{comm_thth'} to shift $\eth$ to the right, up to producing a term of lower order of differentiation every time that $\eth$ is swapped with $\eth'$, and thus iteratively arrive at either zero or at a term proportional to (using~\eqref{comm_etheth'})
\be
 \eth\eth'\Phi_2=\Lambda_0\Phi_2  .
 \label{etheth'Phi_2}
\ee  
To summarize, the components of the covariant derivatives of $\cF$ possessing $(b,s)=(-1,-1)$ are either zero or proportional (via a numerical factor) to 
\be
	\Lambda_0^{k/2}\Phi_2 , 
	\label{DF-1-1}
\ee	
in the latter case $k$ being an even positive integer corresponding to the order of differentiation of the particular covariant derivative being considered.\footnote{This particular power of $\Lambda_0$ also follows from a dimensional argument. Note also that, using \eqref{comm_thth'} and \eqref{comm_etheth'} iteratively (with \eqref{null_cond1}), one can prove  
\be
 \thorn^j\thorn'^j\Phi_2=\left(-\frac{\Lambda_0}{2}\right)^j(j+1)!j!\Phi_2 , \qquad \eth^j\eth'^j\Phi_2=\left(\frac{\Lambda_0}{2}\right)^j(j+1)!j!\Phi_2 ,
\ee    
which generalize \eqref{thotho'Phi_2} and \eqref{etheth'Phi_2} to any integer $j>0$.}  From the previous comments it also follows that $k$-th covariant derivative of~\eqref{F_null_spinor} takes the form
\be
 \nabla_{C_kC'_k}\ldots\nabla_{C_1C'_1}\phi_{AB}=\Lambda_0^{k/2}\phi_{AB}\Upsilon_{C_1C'_1\ldots C_kC'_k} +\ldots \qquad (k \mbox{ even}),
 \label{k_der_null}
\ee
where $\Upsilon_{C_1C'_1\ldots C_kC'_k}$ is a spinor with weights $(0,0)$ and the dots denote terms with weights different from $(-1,-1)$ (see~\eqref{2nd_der_null} below for an explicit example). For certain applications it is also useful to observe that, irrespective of its precise form, the spinor $\Upsilon_{C_1C'_1\ldots C_kC'_k}$ (and thus the full first term on the RHS of~\eqref{k_der_null}) is invariant under the ``Sachs symmetry''\footnote{Mentioned in \cite{GHP,penrosebook1} with reference to unpublished work by Sachs (1961, 1962).} $o_A\mapsto o_A$, $\iota_A\mapsto\iota_A$, $\bar o_{A'}\mapsto \bar\iota_{A'}$, $\bar\iota_{A'}\mapsto-\bar o_{A'}$. By complex conjugation, similar conclusions hold for the components of the covariant derivatives of $\bar\cF$ possessing $(b,s)=(-1,+1)$.\footnote{Although $\bar\phi_{AB}$ (and thus $\bar\cF$) is not invariant under the Sachs symmetry (as opposed to $\phi_{AB}$, cf.~\eqref{F_null_spinor}), it is invariant under a ``dual'' Sachs symmetry $o_A\mapsto -\iota_A$, $\iota_A\mapsto o_A$, $\bar o_{A'}\mapsto\bar o_{A'}$, $\bar\iota_{A'}\mapsto\bar\iota_{A'}$.\label{footn_Sachs}}

More generally, a similar reasoning reveals that, for $s\le-1$, components with $(b,s)=(-1,s)$ contained in the covariant derivatives of $\cF$ will be linear combinations of terms of the form
\be
 \Lambda_0^{(k+s+1)/2}\eth'^{(-s-1)}\Phi_2 ,
 \label{DF-1s}
\ee 
where $k$ is the order of the covariant derivative in question, $k$ and $s$ have opposite parity and $-k-1\le s\le-1$.

Finally, $\phi_{AB}$ does not possess components of b.w.~$-2$, and the components of $\nabla_{CC'}\phi_{AB}$ satisfy~\eqref{b_s_extra} (in addition to \eqref{balanced} and \eqref{b_s_balanced}), therefore (as follows from the comments following~\eqref{b_s_extra}) no non-zero components with $(b,s)=(-2,0)$ are possible in covariant derivatives of $\cF$ or $\bar\cF$ of any order.

For the sake of definiteness, let us demonstrate explicitly the above results for the first two covariant derivatives of \eqref{F_null_spinor}. 
The first derivative reads, thanks to~\eqref{null_cond4},
\be
	\nabla_{CC'}\phi_{AB}=o_Ao_Bo_C\big(\bar o_{C'}\!\!\!\!\underbrace{\thorn'}_{(-2,-1)}\!\!\!\!-\bar\iota_{C'}\!\!\!\underbrace{\eth'}_{(-1,-2)}\!\!\!\big)\Phi_2 ,
	\label{1st_der_null}
\ee	
while the second derivative is
\beqn
	\nabla_{DD'}\nabla_{CC'}\phi_{AB}=o_Ao_Bo_C\Big\{o_D\big[\bar o_{D'}\bar o_{C'}\underbrace{\left(\thorn'^2-\bar\nu\eth'\right)}_{(-3,-1)}-\left(\bar o_{D'}\bar\iota_{C'}+\bar\iota_{D'}\bar o_{C'}\right)\!\!\underbrace{\thorn'\eth'}_{(-2,-2)}\!\!+\bar\iota_{D'}\bar\iota_{C'}\!\!\underbrace{\eth'^2}_{(-1,-3)}\!\!\big] \nonumber \\
 -{}\iota_D\varepsilon_{C'D'}\!\!\!\underbrace{\Lambda_0}_{(-1,-1)}\!\!\!\Big\}\Phi_2 , 
\label{2nd_der_null}
\eeqn
where we have used again \eqref{null_cond4} along with the commutators \eqref{comm_thth'}--\eqref{comm_th'eth}. Below each term we have indicated the weights $(b,s)$ of the component it gives rise to (after acting on $\Phi_2$). Eqs.~\eqref{1st_der_null} and \eqref{2nd_der_null} show that no components with $(b,s)=(-2,0)$ are present, as proven more in general above, and that components with $(b,s)=(-1,s)$ are in agreement with the general results \eqref{DF-1-1} and \eqref{DF-1s}.

\subsubsection{Components of $\nabla^k\bC$ and {$\nabla^k\bSS$}}

Let us now discuss the form of components of covariant derivatives of the Weyl tensor possessing $(b,s)=(-2,0)$ (as noticed in section~\ref{sec_4D_null}, here 
$\nabla^k\BC$ is 1-balanced for $k\ge1$, therefore it cannot contain terms with $b>-2$). In spinorial terms, the first covariant derivative reads (cf. also \cite{HerPraPra17})
\beqn
	\nabla_{EE'}\Psi_{ABCD}=\bar o_{E'}\big(12o_{(A}o_Bo_C\iota_{D)}o_E\!\!\underbrace{\nu\Psi_2}_{(-2,-1)}\!\!-o_{A}o_Bo_Co_D\iota_E\!\!\underbrace{\eth\Psi_4}_{(-2,-1)}\!\!\big) \nonumber \\
	{}+o_{A}o_Bo_Co_Do_E\big(\bar o_{E'}\!\!\!\!\underbrace{\thorn'}_{(-3,-2)}\!\!\!\!-\bar\iota_{E'}\!\!\!\!\underbrace{\eth'}_{(-2,-3)}\!\!\big)\Psi_4 ,
	\label{1st_der_Psi}
\eeqn	
No components with weights $(-2,0)$ are present here, but such terms can be produced by taking further derivatives, as we now explain. Recalling \eqref{null_cond1} and \eqref{null_cond2}, one can observe that all components of \eqref{1st_der_Psi} are expressed only (up to multiplicative constants) in terms of $\nu$ and some of its (second) GHP derivatives. Since $\nu$ has weights $(-2,-1)$, in order to produce (by differentiation) a term with $(-2,0)$ one needs to act on $\nu$ with an equal number of $\thorn$ and $\thorn'$ operators, and with a number of $\eth$ which exceeds precisely by 1 the number of $\eth'$, in any order -- for example one such term is given by (up to reordering of the derivatives) $\thorn'^j\thorn^j\eth'^k\eth^{k+1}\nu$.\footnote{For completeness, it should be observed that if $\eta$ is a component of a certain covariant derivative of $\Psi_{ABCD}$, one additional differentiation also produces terms of the form $\nu\eta$ and $\bar\nu\eta$ (via~\eqref{nonnull_spin_der}, cf. \eqref{covder_compts2}). However, since $\eta$ is 1-balanced and $\thorn\nu=0$ (recall~\eqref{nonnull_cond2}), all such terms are automatically 3-balanced (as defined in \cite{Kuchynkaetal19}, see also Lemma~1 therein), which implies they possess b.w. not greater than $-4$ and can thus be neglected for the purposes of the present discussion.} First, this means that non-zero components with $(-2,0)$ can arise only in covariant derivatives of $\Psi_{ABCD}$ of {\em even} order. Second, using \eqref{comm_thth'}--\eqref{comm_th'eth} one can argue that any such component can be rewritten as a linear combination of terms of the form (for various values of $k$)
\be
 \eth'^k\eth^{k+1}\nu=\eth'^k\eth^{k}\Phi_{22} ,
 \label{terms-2weyl}
\ee
where we have used~\eqref{null_cond2}. Noticing that $(\eth'\eth+\eth\eth')\Phi_{22}=\T \Phi_{22}$, where $\T=2P^2\pa_{\bar\zeta}\pa_\zeta$ is the Laplace operator in the transverse 2-space with metric $2P^{-2}\d\zeta\d\bar\zeta$ (recall~\eqref{rec_metric}, \eqref{W_rec}), and using~\eqref{comm_etheth'}, a linear combination of terms~\eqref{terms-2weyl} can also be expressed equivalently as a linear combination of terms 
\be
 \T^j\Phi_{22} .
 \label{terms-2weyl_Lapl}
\ee
For example, one has $\T^2\Phi_{22}=4(\eth'^2\eth^2-\Lambda_0 \eth'\eth)\Phi_{22}$, $\T^3\Phi_{22}=8(\eth'^3\eth^3-4\Lambda_0 \eth'^2\eth^2+\Lambda_0^2\eth'\eth)\Phi_{22}$, etc..

Since here the Ricci tensor has the form $R_{ab}=\Lambda_0 g_{ab}+2\Phi_{22}\ell_a\ell_b$, the argument given above for the covariant derivatives of the Weyl tensor can easily be adapted to the covariant derivatives of the Ricci tensor, so their components $(-2,0)$ can also only be a linear combination of terms~\eqref{terms-2weyl_Lapl} (and no components with b.w. greater than $-2$ are possible). 


From the fact that the Weyl and Ricci tensor only possess components of b.w. 0 (cf. the comments in the last paragraph of section~\ref{subsec_4D_null_gen}) and $-2$ and their covariant derivatives are 1-balanced, it follows that
any rank-2 tensor constructed from the Riemann tensor and its covariant derivatives of arbitrary order is necessarily symmetric and its component are either proportional to the metric tensor (b.w. 0) or a linear combination of terms~\eqref{terms-2weyl_Lapl} (b.w. $-2$). Thus we have the the following\footnote{A similar result for a different class of spacetimes was obtained in \cite{Gursesetal13,Kuchynkaetal19} (cf. also \cite{KucOrt19}).}
\begin{theorem}[Rank-2 tensors]
 \label{th_T_null}
 In the family of spacetimes~\eqref{rec_metric_spec}, any rank-2 tensor constructed from the Riemann tensor and its covariant derivatives of arbitrary order takes one of the two following equivalent forms
 \be
  \tilde E_{ab}=\lambda_0 g_{ab}+2\ell_a\ell_b\sum_{k=0}^Na_k\T^k\Phi_{22}=\lambda_0 g_{ab}+2\ell_a\ell_b\sum_{k=0}^N\hat a_k\eth'^k\eth^{k}\Phi_{22} , 
	\label{T_null}
 \ee
where $N\in \mathbb{N}$, $\lambda_0$, $a_k$ and $\hat a_k$ are constants, and $\Phi_{22}=\frac{1}{2}\T H$. 
\end{theorem}

\begin{remark}
\label{rem_T_null}
 From the previous discussion it follows that the term of b.w. $-2$ in \eqref{T_null} (i.e., the one proportional to $\ell_a\ell_b$) can only be constructed linearly in the traceless Ricci tensor and the covariant derivatives of arbitrary order of the Riemann tensor. 
\end{remark}

%
%
%
%
%

\providecommand{\href}[2]{#2}\begingroup\raggedright\endgroup

\end{document}